\documentclass[sigconf,nonacm]{acmart}
\usepackage{soul,balance}
\usepackage{booktabs} 
\usepackage{comment}
\usepackage{hyperref} 
\usepackage{graphicx} 	
\usepackage{float,wrapfig} 		    
\usepackage{mathtools,verbatim}      
\usepackage{caption}
\usepackage{subcaption}
\usepackage{multicol}
\usepackage{enumerate}
\usepackage{enumitem}

\usepackage[ruled, vlined,linesnumbered]{algorithm2e}

\begin{document}

\title{
Byzantine Causal Reliable Broadcast (BCRB) with Constant-Size Message Metadata
}

\author{Purv Patel}
\affiliation{%
  \institution{University of Illinois Chicago}
  \city{Chicago}
  \country{USA}
  }
\email{purvp2@uic.edu}

\author{Ajay D. Kshemkalyani}
\affiliation{%
  \institution{University of Illinois Chicago}
  \city{Chicago}
  \country{USA}
  }
\email{ajay@uic.edu}



\begin{abstract}
Asynchronous Byzantine Reliable Broadcast (BRB) is a fundamental primitive that guarantees agreement and validity in distributed systems subject to Byzantine faults, but it lacks ordering guarantees. In this paper, we address Byzantine Causal Reliable Broadcast (BCRB), which builds on BRB to enforce causal message ordering. 
We present a novel BCRB protocol that decouples causal ordering from the BRB layer, achieving constant-size $\mathcal{O}(1)$ message metadata overhead and $\mathcal{O}(n^2)$ communication word complexity as against $\mathcal{O}(n^3)$ communication word complexity of existing protocols; here $n$ is the number of processes.

We present two variants of our protocol: a cryptographic version using a threshold encryption scheme and sequence gating, and its non-cryptographic version. In the cryptographic version, senders broadcast ciphertexts immediately, and decryption shares are piggybacked on out-of-band ACKs, preventing early decryption and front-running. In both versions, causal safety is achieved probabilistically. We evaluate the probability of causal safety violations using a random variable path analysis under independent exponential link delay distributions. We show that both variants satisfy liveness and the probability of weak safety violation is bounded by $\mathcal{O}(f^{-3}\cdot\ln^3 f)$, where $f$ is the upper bound on the number of Byzantine processes, and $f < n/3$ and $f=\mathcal{O}(n)$. Further, for the crypto version, we show that the probability of strong safety violation is bounded by $\mathcal{O}(f^{-1} \cdot \ln^2 f)$.  
We also show how to modify our two protocols to guarantee 100\% weak safety keeping $\mathcal{O}(1)$ message space overhead but with $\mathcal{O}(n^3)$ messages and $\mathcal{O}(n^3)$ communication word complexity.
\end{abstract}


\maketitle

\section{Introduction}
Ensuring consistent transaction order in distributed systems subject to Byzantine failures is essential for applications ranging from decentralized state-machine replication (SMR) to financial ledgers. In these environments, Byzantine processes can observe pending messages, manipulate their delivery order, or inject front-running transactions to extract economic value or compromise consistency. However, probabilistic ordering guarantees are often adequate for non-critical applications such as social networks. While total-order primitives (atomic broadcast) are commonly used to resolve ordering conflicts, they require heavy consensus overhead and use randomization \cite{Cachin2001}. Byzantine Causal Reliable Broadcast (BCRB) offers a lightweight alternative by delivering messages according to Lamport's \textit{happens-before} relationship ($\to$) \cite{LLclock}.
However, scaling BCRB to large networks is severely limited by previous algorithms which append $\mathcal{O}(n)$ metadata to messages and incur $\mathcal{O}(n^3)$ communication complexity, where $n$ is the number of processes 
\cite{DBLP:journals/tcs/AuvolatFRT21,Cachin2001}.

Decoupling causal delivery from the network layer while maintaining a constant $\mathcal{O}(1)$ message metadata overhead presents fundamental challenges in asynchronous systems. It has been shown that (deterministically) achieving both strong causal safety and liveness without using cryptography is impossible in asynchronous Byzantine systems \cite{misra2022detecting, DBLP:journals/tpds/MisraK24,DBLP:journals/pc/MisraK25,DBLP:conf/netys/MisraK22,DBLP:conf/icdcn/MisraK23}. Existing  protocols for BCRB weaken safety to weak safety \cite{DBLP:journals/tcs/AuvolatFRT21,DBLP:journals/tpds/MisraK24}, or accept probabilistic safety guarantees as in \cite{Cachin2001}. Such solutions either couple the application data directly with a costly consensus plane to prevent front-running \cite{Cachin2001}, or incur linear metadata overheads that degrade throughput \cite{minicast,DBLP:journals/tcs/AuvolatFRT21, Cachin2001}.

To resolve this bottleneck, we present a novel BCRB protocol that achieves constant-size $\mathcal{O}(1)$ message metadata. Our protocol layers the causal delivery logic directly atop a standard $\mathcal{O}(n^2)$ messages Byzantine Reliable Broadcast (BRB) primitive \cite{DBLP:journals/iandc/Bracha87}, routing application payloads over authenticated point-to-point links. Causal dependencies are tracked and resolved out-of-band using point-to-point acknowledgments (ACKs) and sequence gating, entirely eliminating vector clocks from the broadcast payloads.

\medskip
\textbf{Contributions:}
\begin{itemize}
    \item 
    \textbf{Cryptographic and Non-Cryptographic Variants:} We present a cryptographic version of the protocol (Algorithm~\ref{alg:bcrb_algorithm}) that secures payload privacy using threshold decryption. 
    Its non-cryptographic version is Algorithm~\ref{alg:bcrb_non_crypto_algorithm}. 
    \item \textbf{Constant Message Metadata Space ($\mathcal{O}(1)$):} Message metadata is reduced to just $\mathcal{O}(1)$ number of $\mathcal{O}(1)$ sized fields, eliminating linear vector overheads.
    \item \textbf{Communication Word Complexity ($\mathcal{O}(n^2)$):} While existing Byzantine causal broadcast protocols require $\mathcal{O}(n^3)$ communication word complexity \cite{DBLP:journals/tcs/AuvolatFRT21,Cachin2001}, our protocols leverage point-to-point ACKs to sequence dependencies, maintaining a total communication word  complexity of $\mathcal{O}(n^2)$ per causal broadcast.
    \item \textbf{Probabilistic Analysis}: We model link propagation delays as independent exponential variables. We formally prove that the probability of causal weak safety violations decays at a rate of $\mathcal{O}(f^{-3} \cdot \ln^3 f)$, ensuring high causal integrity in practice. Here $f$ is the upper bound on the number of Byzantine processes, $f<n/3$ (the BRB resilience bound) and $f=\mathcal{O}(n)$. We also prove that for the cryptographic version of our protocol, the probability of causal strong safety violations is $\mathcal{O}(f^{-1} \cdot \ln^2 f)$. 
    \item \textbf{Weak Safety Guarantee Extensions:} We propose Algorithm~\ref{alg:bcrb_algorithm_n3}, a  variant of Algorithm~\ref{alg:bcrb_algorithm}, that replaces direct ACK broadcasts with a BRB of the ACKs. It achieves 100\%  causal weak safety, liveness, and a probability of causal strong safety violations $\mathcal{O}(f^{-1} \cdot \ln^2 f)$. It has $\mathcal{O}(1)$ message space overhead at the cost of a higher $\mathcal{O}(n^3)$ communication word complexity. Algorithm~\ref{alg:bcrb_non_crypto_algorithm_n3} is its non-crypto variant; it lacks a bound on probability of causal strong safety violations.
\end{itemize}
All our algorithms are {\em throughput-scalable}, i.e., throughput and rate of sending are not limited by the latency of messages. 
We do not wait for the previous BRB to be locally delivered before sending the next BRB.

\section{System Model and Background}
\label{sec:background}

\subsection{System Model}
We assume an asynchronous distributed system of $n$ processes, denoted by $p_1, p_2, \dots, p_n$. The system is subject to Byzantine failures. Any $f$ processes can behave arbitrarily by dropping messages, forging causal histories, or colluding. The number of Byzantine processes is bounded by $f < n/3$ (the BRB resilience bound) \cite{DBLP:journals/iandc/Bracha87,DBLP:journals/jacm/BrachaT85} and our probability analysis further assumes $f=\mathcal{O}(n)$. Processes communicate via reliable and authenticated point-to-point FIFO channels. We assume an underlying Byzantine Reliable Broadcast (BRB) primitive (e.g., Bracha's BRB \cite{DBLP:journals/iandc/Bracha87}) that satisfies Validity, Agreement, and Integrity.
\begin{itemize}
     \item \textbf{Validity}: If a correct process broadcasts $m$, all correct processes eventually deliver $m$.
     \item \textbf{Agreement}: If a correct process delivers $m$, all correct processes eventually deliver $m$.
     \item \textbf{Integrity}: A message $m$ is delivered at most once by each correct process, and if the sender is correct, then only if it was broadcasted by the sender.
 \end{itemize}

We define Lamport's happens-before relation $\to$ \cite{LLclock} over the set of events in the system. For any two events $e_1, e_2$, we have $e_1 \to e_2$ if: (1) they occur at the same process and $e_1$ precedes $e_2$ in local execution; (2) $e_1$ is the broadcast event of a message and $e_2$ is the delivery event of that message; or (3) there exists an event $e_3$ such that $e_1 \to e_3$ and $e_3 \to e_2$.
For the cryptographic version of the protocol, we assume an $(n-f, n)$ threshold decryption scheme (e.g., Shoup's threshold cryptosystem \cite{shoup2000practical}), where a ciphertext can only be decrypted once at least $n-f$ valid decryption shares from distinct processes are collected.


\subsection{Layered Broadcast Primitives}
Our protocol separates causal ordering logic from reliable delivery by layering the Byzantine Causal Reliable Broadcast (BCRB) layer directly atop an underlying Byzantine Reliable Broadcast (BRB) layer (e.g., Imbs-Raynal BRB \cite{DBLP:journals/ppl/ImbsR16} or Bracha's BRB \cite{DBLP:journals/iandc/Bracha87}) requiring $\mathcal{O}(n^2)$ messages. Note, while our protocol design is general, the probabilistic causal safety analysis presented in Section~\ref{sec:proof} is specifically valid when utilizing Bracha's BRB algorithm \cite{DBLP:journals/iandc/Bracha87}, as it relies on Bracha's specific quorum amplification and threshold properties. 

We define the four primitives of this layered architecture and specify their parameters:
\begin{itemize}
    \item \textbf{Underlying Layer (BRB)}:
    \begin{itemize}
        \item $\text{\texttt{brb\_broadcast}}(sender, sn, content)$: Invoked by the BCRB layer of the sending process $p_{sender}$ to reliably broadcast $content$ (representing either a ciphertext or plaintext) with a sequence number $sn$.
        \item $\text{\texttt{brb\_deliver}}(sender, sn, content)$: Upcall from the BRB layer to the BCRB layer when $content$ from $p_{sender}$ with sequence number $sn$ is delivered. The BRB layer guarantees Validity, Agreement, and Integrity.
    \end{itemize}
    \item \textbf{Causal Layer (BCRB)}:
    \begin{itemize}
        \item $\text{\texttt{bcrb\_broadcast}}(payload)$: Invoked by the application layer at process $p_i$ to causally broadcast a plaintext $payload$.
        \item $\text{\texttt{bcrb\_deliver}}(sender, sn, plaintext)$: Upcall from the BCRB layer to the application layer to deliver the ordered $plaintext$ payload originally sent by $p_{sender}$ with sequence number $sn$.
    \end{itemize}
\end{itemize}
The layering operates as follows: $\text{\texttt{bcrb\_\-broad\-cast}}\-(payload)$ encapsulates and encrypts the payload, triggering $\text{\texttt{brb\_\-broad\-cast}}\-(i, sn, ciphertext)$ on the ciphertext. When $\text{\texttt{brb\_deliver}}(sender,\- sn, ciphertext)$ occurs, the BCRB layer buffers the message and verifies FIFO sequence numbers, sequence blocks, and decryption shares. Once verified, the BCRB layer decrypts the payload and triggers $\text{\texttt{bcrb\_deliver}}(sender, sn, plaintext)$ to the application. For the non-cryptographic version, the layering follows the same flow but bypasses the encryption and decryption steps.

\subsection{Safety and Liveness Definitions}
\begin{definition}
The happens before relation $\to$ on application messages consists of the following rules:
\begin{enumerate}
\item If $p_i$ \texttt{bcrb\_broadcast}($m$) or \texttt{bcrb\_deliver}($m$) before \texttt{bcrb\_\-broadcast}($m'$), then $m \to m'$.
\item If $m \to m'$ and $m' \to m''$, then $m \to m''$.
\end{enumerate}
\end{definition}
\begin{itemize}
    \item \textbf{Strong Causal Safety}: If $m_1 \to m_2$ then no correct process triggers \texttt{bcrb\_deliver}($m_2$) before \texttt{bcrb\_deliver}($m_1$) \cite{DBLP:journals/tpds/MisraK24,DBLP:conf/netys/MisraK22,DBLP:conf/icdcn/MisraK23}.
    \item \textbf{Weak Causal Safety}: If $m_1$ $\to$ $m_2$ and the causal chain from $m_1$ to $m_2$ passes exclusively through correct processes, then no correct process triggers \texttt{bcrb\_deliver}($m_2$) before \texttt{bcrb\_deliver}($m_1$) \cite{DBLP:journals/tpds/MisraK24,DBLP:conf/netys/MisraK22,DBLP:conf/icdcn/MisraK23}.
    \item \textbf{Liveness}: Every message broadcasted via \texttt{bcrb\_broadcast}\-(payload) by a correct process is eventually delivered via \texttt{bcrb\_\-deliver}(payload) at all correct processes. Liveness is same as Validity.
\end{itemize}

Additionally, to reason with our algorithms and executions where Byzantine processes may skip the BCRB layer and directly invoke the underlying BRB layer, we define the happens-before relation $\to_{brb}$ over the messages broadcast via the BRB layer. 
\begin{definition}
The happens before relation $\to_{brb}$ on messages broadcast by the BRB layer consists of the following rules:
\begin{enumerate}
\item If $p_i$ \texttt{brb\_broadcast}($m$) or \texttt{brb\_deliver}($m$) before \texttt{brb\_\-broadcast}($m'$), then $m \to_{brb} m'$.
\item If $m \to_{brb} m'$ and $m' \to_{brb} m''$, then $m \to_{brb} m''$.
\end{enumerate}
\end{definition}

\begin{table*}[htbp]
\centering
\caption{Comparative Analysis of Byzantine Causal Broadcast Protocols}
\label{tab:comparison}
{\small
\begin{tabular}{lcccccc}
\hline
\textbf{Property} & \textbf{Auvolat et al. \cite{DBLP:journals/tcs/AuvolatFRT21}} & \textbf{Cachin et al. \cite{Cachin2001}} & \textbf{Algorithm~\ref{alg:bcrb_algorithm}} & \textbf{Algorithm~\ref{alg:bcrb_non_crypto_algorithm}}  &\textbf{Algorithm~\ref{alg:bcrb_algorithm_n3}} &  \textbf{Algorithm~\ref{alg:bcrb_non_crypto_algorithm_n3}} 
\\
\hline
Deterministic & Yes & No (randomized) & Yes & Yes & Yes & Yes 
\\
Crypto & No & Yes & Yes & No & Yes & No 
\\
Message Size Overhead & $\mathcal{O}(n)$ 
 & $\mathcal{O}(n)$ (Encrypt. Hdrs) 
 & $\mathcal{O}(1)$  & $\mathcal{O}(1)$  & $\mathcal{O}(1)$ & $\mathcal{O}(1)$ 
 \\
Communic. Complexity & $\mathcal{O}(n^3)$ & $\mathcal{O}(n^3)$ & $\mathcal{O}(n^2)$ & $\mathcal{O}(n^2)$ & $\mathcal{O}(n^3)$ & $\mathcal{O}(n^3)$ 
\\
Front-Running Protection & No & Yes & Yes & No & Yes & No 
\\
Throughput-scalable & Yes & No & Yes & Yes & Yes & Yes 
\\
$P$(weak safety violation) & 0 & $\epsilon > 0$ & $\mathcal{O}(f^{-3} \cdot \ln^3 f)$ & $\mathcal{O}(f^{-3} \cdot \ln^3 f)$ & 0 & 0 
\\
$P$(strong safety violation) & high; front-running $+$  & $\epsilon > 0$ & $\mathcal{O}(f^{-1} \cdot \ln^2 f)$ & front-running  & $\mathcal{O}(f^{-1} \cdot \ln^2 f)$ & front-running 
\\
 & fake causal barriers & & & \\
\hline
\end{tabular}
}
\end{table*}

\textbf{Problem Definition.} BCRB must satisfy Validity, Agreement, Integrity, (Weak and/or Strong) Safety, and Liveness. 

As we layer BCRB over BRB, Validity, Agreement, and Integrity which are provided by BRB need to be satisfied by the BCRB layer, besides 
proving Weak or Strong Safety, and Liveness.

\subsection{Byzantine Causal Broadcast}
Misra et al. \cite{DBLP:journals/tpds/MisraK24,DBLP:journals/pc/MisraK25,DBLP:conf/netys/MisraK22} showed that it is impossible to provide both strong safety and liveness for causal ordering without using cryptography, but weak safety and liveness can be provided. It was formally proved in \cite{misra2022causal} that Bracha's BRB \cite{DBLP:journals/iandc/Bracha87} does not satisfy even the weak safety property.

Byzantine causal broadcast algorithm by Auvolat et al. \cite{DBLP:journals/tcs/AuvolatFRT21} enforces causal order by attaching a ``causal barrier" (predecessor message IDs) to every message. This incurs an $\mathcal{O}(n)$ space overhead in application messages and runs directly over an $\mathcal{O}(n^2)$ messages BRB primitive, generating $\mathcal{O}(n^3)$ message communication complexity. 
Cachin et al. \cite{Cachin2001} proposed a secure causal atomic broadcast protocol using threshold encryption. A sender broadcasts the encrypted payload via an atomic broadcast channel. The use of atomic broadcast leads to a probabilistic, randomized (non-deterministic) solution. Once the ciphertext is totally ordered, processes exchange decryption shares to reveal the plaintext. Although this prevents front-running, it tightly couples the data and control planes, forcing large application payloads to be processed by the expensive total-order consensus layer ($\mathcal{O}(n^3)$ communication complexity). Our protocols separate the data and control planes, routing payloads via $\mathcal{O}(n^2)$ BRB and exchanging decryption shares out-of-band via point-to-point ACKs. A comparison with our protocols is given in Table~\ref{tab:comparison}.  Auvolat et al. \cite{DBLP:journals/tcs/AuvolatFRT21} provides liveness and weak safety. Cachin et al. \cite{Cachin2001} provides liveness, and being randomized, provides weak safety and strong safety with high probability.

\section{Algorithm~\ref{alg:bcrb_algorithm}: Cryptographic Version}
\label{sec:algo}
The cryptographic version of our BCRB protocol integrates an $(n-f, n)$ threshold decryption scheme and a sequence gating mechanism. Senders encrypt payloads and broadcast immediately. Decryption shares are piggybacked on out-of-band ACKs. 

\begin{algorithm*}[tp]
\small
\SetKwComment{Comment}{// \textit{}}{}
\SetKwFunction{BCRBBroadcast}{bcrb\_broadcast}
\SetKwFunction{BRBBroadcast}{brb\_broadcast}
\SetKwFunction{CheckDelivery}{check\_delivery}
\SetKwFunction{BRBDeliver}{brb\_deliver}
\SetKwFunction{ReceiveACK}{receive ACK}
\SetKwProg{Fn}{procedure}{:}{}
\SetKwProg{On}{upon}{:}{}

\caption{Byzantine Causal Reliable Broadcast (BCRB) with Constant-Size Messages (Process $i$)}
\label{alg:bcrb_algorithm}

\begin{multicols}{2}
\SetInd{0.2em}{0.5em}

\textbf{state variables:} \\
  $local\_sn \gets 0$ {\footnotesize \Comment*{Seq. num. for process $i$'s broadcasts}}
  $\mathbf{V}_i \gets [0, \dots, 0]$ {\footnotesize \Comment*{Vec. of BCRB-delivered seq. nos.}}
  $pending \gets \emptyset$ {\footnotesize \Comment*{processed BRB-delivery, pending BCRB delivery}}
  $shares \gets \text{Array of } \emptyset$ {\footnotesize \Comment*{Maps msg ID to decryption shares}}
  $G=(N,E) \gets (\emptyset,\emptyset)$ {\footnotesize \Comment*{graph of dependencies between msg IDs}}

\BlankLine
\Fn{\BCRBBroadcast{payload}}{
    $local\_sn \gets local\_sn + 1$\;
    $ciphertext \gets \text{encrypt}(payload, \mathit{PK})$\;
    \BRBBroadcast{$(i, local\_sn, ciphertext)$}\;
}

\BlankLine
\On{\BRBDeliver{j, sn, ciphertext}}{
    \textbf{wait until } $M (= (j,sn'))$ for all $sn' < sn$ are brb-delivered and execution of brb-deliver is completed\;
    $share \gets \text{dec\_share}(ciphertext, \mathit{SK}_i)$\;
    $next\_sn \gets local\_sn + 1$\;
    $h \gets hash(ciphertext)$\;
    $\text{broadcast } ACK((j, sn), next\_sn, share, i, h) \text{ to all}$\;
    \If{$(j,sn)^{temp}$ \text{ exists }}{
        relabel it as $(j,sn)$; $(j,sn).ts \gets$ local physical clock timestamp\;
    }
    \Else{
        add $(j,sn)$ to $N$; $(j,sn).ts \gets$ local physical clock timestamp\;
    }
    \If{$(j, sn-1) \in N$}{
        add edge $((j,sn),(j,sn-1))$ to $E$\;
    }
    $pending \gets pending \cup \{((j, sn), ciphertext)\}$\;
    \CheckDelivery{}\;
}

\columnbreak

\BlankLine
\On{\ReceiveACK{M, next\_sn, share, k, h}}{
    \textbf{wait until } $M (= (l,s))$ and all $M' = (l,s')$ where $s' < s$ are brb-delivered and execution of brb-deliver is completed\;
    \If{$M \in N  \land h = hash(M.ciphertext) \land V_i[k] < next\_sn$}{
        $shares[M] \gets shares[M] \cup \{share\}$\;
        \If{$(k, next\_sn) \in N \land (k, next\_sn).ts >(l,s).ts$}{
            add $((k,next\_sn),(l,s))$ to $E$\;
        }
        \ElseIf{$(k,next\_sn) \not\in N$}{
            add $((k,next\_sn)^{temp},(l,s))$ to $E$\;
        }
        \CheckDelivery{}\;
    }
}

\BlankLine
\Fn{\CheckDelivery{}}{
    $progress \gets \text{True}$\;
    \While{$progress = \text{True}$}{
        $progress \gets \text{False}$\;
        \ForEach{$m = (M = (j, sn), ciphertext) \in pending$}{
            $fifo\_ok \gets (\mathbf{V}_i[j] = sn - 1)$\;
            $quorum\_ok \gets (|shares[M]| \ge n - f)$\;
            
            \If{$fifo\_ok \land quorum\_ok \land (j,sn)  \text{ has no outgoing edge in } E$}{
                $pending \gets pending \setminus \{m\}$\; delete $M$ from $N$ and all incident edges in $E$\; 
                $plaintext \gets \text{decrypt}(ciphertext, shares[M], \mathit{PK})$\;
                $\mathbf{V}_i[j] \gets sn$\;
                $\text{bcrb\_deliver}(j, sn, plaintext)$\;
                $progress \gets \text{True}$\;
                \textbf{break}\;
            }
        }
    }
}

\end{multicols}
\end{algorithm*}

\subsubsection*{bcrb\_broadcast(payload)}
When the application layer broadcasts a payload, the sender increments its $local\_sn$. The payload is encrypted under the system public key $\mathit{PK}$ to obtain a $ciphertext$. A message $m = (i, local\_sn, ciphertext)$ is created and immediately broadcasted via the underlying BRB layer. The message carries no vector clocks, ensuring $\mathcal{O}(1)$ message overhead.

\subsubsection*{upon brb\_deliver(j, sn, ciphertext)}
Upon BRB-delivering a message $(j, sn, ciphertext)$, a process waits until all preceding messages $M (= (j, sn'))$ for $sn' < sn$ are BRB-delivered to preserve FIFO order. It then computes a cryptographic decryption share $share$ using its secret key share $\mathit{SK}_i$, determines its next sequence number $next\_sn \gets local\_sn + 1$, and computes the payload hash $h \gets \text{hash}(ciphertext)$. It broadcasts $ACK((j, sn), next\_sn, share, i, h)$ to all processes.

The process then updates the local dependency graph $G = (N, E)$ to represent the causal history: if a temporary node $(j, sn)^{temp}$ already exists (created by a prior ACK delivery), it relabels it as the concrete node $(j, sn)$; otherwise, it adds $(j, sn)$ directly to the node set $N$. It assigns the local physical clock value as the timestamp $(j,sn).ts$. To enforce FIFO order, if $(j, sn-1)$ is in $N$, it adds the edge $((j, sn), (j, sn-1))$ to $E$. It adds the message to the \texttt{pending} set. 
It then calls \texttt{check\_delivery()}.

\subsubsection*{upon receive ACK(M, next\_sn, share, k, h)}
When process $p_i$ receives an ACK for message $M (= (l, s))$ from process $k$ declaring $next\_sn$, $share$, and hash $h$, it waits until $M$ and all prior messages $M' = (l, s')$ (where $s' < s$) from process $l$ are BRB-delivered locally. It further verifies that the message $M$ is registered in $N$, the hash matches ($h = \text{hash}(M.ciphertext)$), and the sequence number is fresh ($V_i[k] < next\_sn$). If valid, it records $share$ in \texttt{shares}$[M]$ and updates the dependency graph $G = (N, E)$ as follows:
\begin{itemize}
    \item If the message node $(k, next\_sn)$ is already present in $N$ and $(k, next\_sn).ts > (l,s).ts$ it adds the edge $((k, next\_sn), (l, s))$ to $E$. This ensures deadlock is avoided.
    \item If the message node $(k, next\_sn)$ is not yet in $N$, it adds the edge $((k, next\_sn)^{temp}, (l, s))$ to $E$ using a temporary node.
\end{itemize}
It then calls \texttt{check\_delivery()}.

\subsubsection*{check\_delivery()}
This procedure evaluates the \texttt{pending} set. A message $m = (M=(j, sn), ciphertext)$ is delivered to the application if:
\begin{enumerate}
    \item FIFO order is satisfied locally for the sender $j$ ($\mathbf{V}_i[j] = \text{sn} - 1$).
    \item It has gathered at least $n - f$ valid decryption shares ($|\text{\texttt{shares}}[M]| \ge n - f$).
    \item The node $(j, sn)$ has no outgoing edge in $E$ (i.e., there are no active causal dependencies blocking its delivery).
\end{enumerate}
If satisfied, the process removes $m$ from the \texttt{pending} set, deletes node $(j, sn)$ and all its incident edges from the graph $G$, decrypts the ciphertext using the gathered shares, triggers \texttt{bcrb\_deliver} on the plaintext, updates $\mathbf{V}_i[j] \gets sn$, and loops to check if further pending messages can now be delivered.

\section{Algorithm~\ref{alg:bcrb_non_crypto_algorithm}: Non-Cryptographic Version}
\label{sec:non_crypto}
Algorithm~\ref{alg:bcrb_non_crypto_algorithm} pseudo-code is same as Algorithm~\ref{alg:bcrb_algorithm}, except (1) ciphertext is identical to plaintext, i.e., no encryption/decryption, and (2) decryption share is set to the sender process ID. Without threshold decryption, a Byzantine process can read plaintexts in transit and attempt front-running and not wait for $n-f$ ACKs before triggering \texttt{bcrb\_deliver}.  

\begin{algorithm}[th]
\caption{Non-Cryptographic Byzantine Causal Reliable Broadcast (BCRB) (Process $i$)}
\label{alg:bcrb_non_crypto_algorithm}
\small
Same as Algorithm~\ref{alg:bcrb_algorithm} except: (1) ciphertext is identical to plaintext, i.e., encryption/decryption are idempotent operations, and (2) decryption share is set to the sender process ID. \\
\end{algorithm}

\section{Correctness Proof for Algorithms~\ref{alg:bcrb_algorithm},~\ref{alg:bcrb_non_crypto_algorithm}}
\label{sec:proof}

\subsection{Background: Bracha's Reliable Broadcast}
To ground our correctness and probability analysis, we review Bracha's Byzantine Reliable Broadcast (BRB) algorithm \cite{DBLP:journals/iandc/Bracha87}, which serves as our underlying BRB primitive. Bracha's protocol uses three types of messages: \texttt{init} (for \texttt{initial}), \texttt{echo}, and \texttt{ready}. A broadcast is initiated by a sender broadcasting \texttt{init}. The correct processes transition through three steps based on threshold quorums:
\begin{itemize}
    \item \textbf{Step 1}: A process waits to receive one \texttt{init} message from the sender. 
    It then broadcasts \texttt{echo} to all.
    \item \textbf{Step 2}: A process waits to receive $\frac{n+f}{2}$ \texttt{echo} messages or $f+1$ \texttt{ready} messages. It then broadcasts \texttt{ready} to all.
    \item \textbf{Step 3}: A process waits to receive $2f+1$ \texttt{ready} messages, after which it delivers (accepts) the payload.
\end{itemize}
These step thresholds ensure that even if a Byzantine sender attempts to equivocate, no two correct processes can deliver different payloads, and if any correct process delivers a payload, all correct processes eventually deliver it. That is, Validity, Agreement, and Integrity are satisfied.

\subsection{Correctness Proofs --- Weak Safety}

\begin{theorem}[Byzantine Front-Running Mitigation]
\label{th:front-running}
Under the $(n-f, n)$ threshold decryption scheme, a Byzantine process cannot decrypt or learn the plaintext content of any message $m$ broadcasted by a correct process before it has been BRB-delivered by at least $n - 2f$ correct processes.
\end{theorem}
\begin{proof}
Reconstructing the plaintext requires at least $n-f$ distinct decryption shares. Correct processes only compute and broadcast their decryption share $share$ of a message $m$ after they have BRB-delivered $m$. Suppose a Byzantine process attempts to decrypt $m$ prematurely. The Byzantine process controls at most $f$ shares. To decrypt, it must obtain at least $(n-f) - f = n-2f$ decryption shares from correct processes. Since correct processes only broadcast their share upon BRB-delivering $m$, the Byzantine process cannot obtain these shares until at least $n-2f$ correct processes have BRB-delivered $m$. Under standard $n \ge 3f + 1$ resilience, this requires at least $f+1$ correct processes to have BRB-delivered $m$. Thus, the Byzantine process cannot decrypt or front-run the message contents early, thereby mitigating the impact of front-running attacks.
\end{proof}

We first analyze the probability of a weak causal safety violation assuming no deadlocks arise and no edge is deleted in the dependency graph $G$ for deadlock avoidance. Then in Section~\ref{sec:deadlock} we analyze the impact of an edge deletion on this probability. 

\begin{theorem}[Probabilistic (Weak) Causal Safety]
Let process $p_i$ broadcast message $m_1$, and correct process $p_k$ broadcast message $m_2$ after delivering $m_1$. The probability that any correct process $p_j$ delivers $m_2$ before $m_1$ is 
\[
P(\text{CO violation at } p_j) = 
\]
\[
P(B < A) \times [1 - (1 - P(B < A')) \times (1 - P(D < C))]
\]
where $A$, $B$, $C$, and $D$ are defined in the body of the proof.
\label{th:wsformulation}
\end{theorem}
\begin{proof}
Let $p_i$ broadcast $m_1$ and $p_k$ broadcast $m_2$ after delivering $m_1$. The sender $p_k$ must have BCRB-delivered $m_1$ prior to broadcasting $m_2$. 
Under optimal $n \ge 3f + 1$ resilience, there are at least $f + 1$ correct processes that have initiated the cascade of $m_1$ by broadcasting their \texttt{READY($m_1$)} messages.
To model the probability of a causal ordering (weak safety) violation at a correct process $p_j$ (delivering $m_2$ before $m_1$), we define the critical path delay variables. Let each network link delay be an independent exponential random variable with rate $\lambda$ (mean delay $1/\lambda$). Let $T_{\text{cascade}}(M, \text{type}, N, \text{max}) = \max(X_1, \dots, X_N)$ represent the delay to complete a quorum phase, where $N$ correct processes send control messages of a given type and all of them arrive at the destination. We let $T_i^{(h)}$ represent the $h$-th independent realization (or instantiation) of these quorum phase distributions for message $M$, capturing independent runs of the same protocol phase.
We next define the critical path time durations for the delivery of $m_1$ and $m_2$ at $p_j$.

As $p_k$ has BRB-delivered $m_1$, it has received at least $f+1$ READY($m_1$) messages from correct processes, which must have been sent. Start measuring from this sending point in time that triggers a cascade in Bracha's BRB.
\begin{itemize}
    \item \textbf{Path A} (BCRB-delivery of $m_1$ to $p_j$, worst case):
    \[
    A = T_1^{(1)} + T_2^{(1)} + T_3^{(1)}
    \]
    where:
    \begin{itemize}
        \item $T_1^{(1)} \sim T_{\text{cascade}}(m_1, \text{READY}, f+1, \text{max})$ represents the time 
        for the $f+1$ READY messages sent by those correct processes to reach other $f$ correct processes (so that then can then amplify the cascade by themselves broadcasting READY).
        \item $T_2^{(1)} \sim T_{\text{cascade}}(m_1, \text{READY}, f+1, \text{max})$ represents the second phase of Bracha's BRB READY cascade. This is the time required for these $f+1$ \texttt{READY($m_1$)} messages broadcast at the end of $T_1^{(1)}$ to propagate and be received (besides the $f+1$ READY broadcast at the start of $T_1^{(1)}$) by a quorum of $2f+1$ correct processes, causing them to broadcast \texttt{ACK($m_1$)}. 
        \item $T_3^{(1)} \sim T_{\text{cascade}}(m_1, \text{ACK}, 2f+1, \text{max})$ represents the time required for the $2f+1$ ACK messages from correct processes to reach the destination process $p_j$ 
        to satisfy the BCRB delivery quorum check of $n-f$ ACKs.
    \end{itemize}

    \item \textbf{Path A'} (BRB-delivery of $m_1$ to $p_j$, worst-case):
    \[
    A' = T^{(1)}_1 + T^{(1)}_2
    \]
    
    \item \textbf{Path B} (BCRB-delivery of $m_2$ to $p_j$, best case):
    \[
    B = T_1^{(2)} + T_2^{(2)} + Y_1(m_2, \text{INIT}, 1, -) + T_3^{(2)} + T_4^{(2)} + T_5^{(2)}
    \]
    where:
    \begin{itemize}
        \item $T_1^{(2)} \sim T_{\text{cascade}}(m_1, \text{READY}, f+1, \text{max})$ is the first cascade phase of $m_1$ (time for the READY messages sent at the start of $T_1^{(1)}$ to reach $p_k$). We assume the READY from Byzantine processes have already reached $p_k$. $T_1^{(2)}$ has same duration as $T_1^{(1)}$.
        \item $T_2^{(2)} \sim T_{\text{cascade}}(m_1, \text{ACK}, f+1, \text{max})$ is the ACK for $m_1$ propagation time from $f+1$ correct processes to the sender of $m_2$ (process $p_k$) so that $p_k$ delivers $m_1$ after which it broadcasts $m_2$. Assume the ACK from Byzantine processes have already reached $p_k$.
        \item $Y_1(m_2, \text{INIT}, 1, -) \sim \text{Exp}(\lambda)$ is the direct link propagation delay for $m_2$'s initial broadcast from $p_k$ to the correct processes.
        \item $T_3^{(2)} \sim T_{\text{cascade}}(m_2, \text{ECHO}, 2f+1, \text{max})$ is the ECHO phase of $m_2$ in Bracha's BRB.
        \item $T_4^{(2)} \sim T_{\text{cascade}}(m_2, \text{READY}, f+1, \text{max})$ is the READY phase of $m_2$ from correct processes. Assume $f$ Byzantine processes have already sent the READY.
        \item $T_5^{(2)} \sim T_{\text{cascade}}(m_2, \text{ACK}, f+1, \text{max})$ is the time for ACK phase of $m_2$ from correct processes (to reach $p_j$). Assume $f$ Byzantine processes have already sent the ACK. $p_j$ now satisfies BCRB delivery quorum check.
    \end{itemize}
    \item \textbf{Path B'} (BRB-delivery of $m_2$ to $p_j$, best case):
    \[
    B' = T_1^{(2)} + T_2^{(2)} + Y_1(m_2, \text{INIT}, 1, -) + T_3^{(2)} + T_4^{(2)}
    \]
\end{itemize}
Counting from the point where BRB delivery of $m_1$ occurs at $p_k$:
\begin{itemize}
    \item \textbf{Path C} (delivery of the direct ACK of $m_1$ from $p_k$ to $p_j$):
    \[
    C = X_1(m_1, \text{ACK}, 1, -) \sim \text{Exp}(\lambda)
    \]
    \item \textbf{Path D} (BCRB-delivery of $m_2$ to $p_j$, best case):
    \[
     D = T_2^{(3)} + Y_1(m_2, \text{INIT}, 1, -) + T_3^{(3)} + T_4^{(3)} + T_5^{(3)}
    \]
    where $T_2^{(3)} \sim T_{\text{cascade}}(m_1, \text{ACK}, f+1, \text{max})$, $T_3^{(3)} \sim \-T_{\text{cascade}}\-(m_2, \text{ECHO}, 2f+1, \text{max})$, $T_4^{(3)} \sim T_{\text{cascade}}(m_2, \text{READY}, f+1, \text{max})$, and $T_5^{(3)} \sim T_{\text{cascade}}(m_2, \text{ACK}, f+1, \text{max})$.
\end{itemize}
Under the approximation assumption that each phase ends at all correct processes at the same time, the sequential phases are independent. Observe from Algorithm~\ref{alg:bcrb_algorithm} that a causal ordering violation occurs at $p_j$ if $B < A$ and it does not happen that $A' < B$ and $C < D$. (When $A'<B$ and $C<D$, $p_j$ adds the dependency edge $(m_2,m_1)$ in $G$, 
and $m_1$ has been BRB-delivered to $p_j$.) Thus when $m_2$ from $p_k$ is BRB-delivered to $p_j$, its BCRB-delivery is not blocked by $m_1$. Note, since $p_i$ is correct, its INIT messages follow FIFO order, and all correct processes will broadcast ECHO and the READY messages in expected FIFO order. So we assume that earlier than $m_1$ messages $p_i$ BRB-broadcast are also expected to have BRB-delivered to $p_j$ before $m_1$.  
Since the phases are independent, the probability of the joint event of CO violation is:
\[
P(B < A) \times [1 - P(A' < B) \times P(C < D)] 
\]
\[
= P(B < A) \times [1 - (1 - P(B < A')) \times (1 - P(D < C))]
\]

We later consider the impact of deadlock avoidance by considering $P(B'<A')$ in Theorem~\ref{th:wsdeadlock} (Section~\ref{sec:deadlock}).
\end{proof}

\subsection{Analytic Evaluation of Probabilities}
\label{sec:probabs}
To simplify notation, we relabel the terms $T^{(1)}_1$, $T^{(2)}_1$, $T^{(3)}_1$ in path $A$ as $X_1, X_2, X_3$, respectively, the terms $T^{(1)}_2$, $T^{(2)}_2$, $T^{(3)}_2$, $T^{(4)}_2$, $T^{(5)}_2$, and $Y_1$ in path $B$ as $Y_1,Y_2,Y_3,Y_4,Y_5$, and $Y$, respectively.

\begin{theorem}[$P(B<A), P(B<A'), P(B'<A')$]
\(P(B<A) = \mathcal{O}\left(\frac{\ln ^{2}f}{f^{2}}\right) \), 
\(P(B<A') = \mathcal{O}\left(\frac{\ln f}{f^{3}}\right)\),
\(P(B'<A') = \mathcal{O}\left(\frac{\ln f}{f^{2}}\right)\)
\label{th:bla}
\end{theorem}
\begin{proof}
To compute $P(B<A)$,
the exact probability that \(Y + Y_1 + Y_2 + Y_3 + Y_4 + Y_5 < X_1 + X_2 + X_3\) can be solved using Rényi’s representation for exponential order statistics.

\begin{table*}[htbp]
\centering
\caption{Comparative Analysis of Features of the Probability derivations.}
\label{tab:probcomparison}
{\small
\begin{tabular}{lccccc}
\hline
System State & Dominant RHS Growth & Dominant LHS Growth & Net Deterministic Drag & Positive Poles (\(k\)) & Asymptotic Probability
\\
\hline
$P(B<A)$ & \(3\ln f\) & \(5\ln f\) &\(2\ln f\) & \(3\) & \(\mathcal{O}\left(\frac{\ln^2 f}{f^2}\right)\)
\\
$P(B<A')$; Remove \(X_{3}\) & \(2\ln f\) & \(5\ln f\) & \(3\ln f\) & \(2\) & \(\mathcal{O}\left(\frac{\ln f}{f^3}\right)\) 
\\
$P(B'<A')$; Remove \(X_{3}\),\(Y_{5}\) & \(2\ln f\) & \(4\ln f\) & \(2\ln f\) & \(2\) & \(\mathcal{O}\left(\frac{\ln f}{f^2}\right)\)
\\
\hline
\end{tabular}
}
\end{table*}

Using Rényi’s representation, the maximum of \(K\) independent identically distributed exponential variables with rate \(\lambda \) is identically distributed to a weighted sum of independent standard exponentials with rates scaling from \(1\lambda\) to \(K\lambda\). Setting \(\lambda = 1\) without loss of generality, we can group the total number of overlapping independent exponential variables (multiplicities) on both sides:

Right-Hand Side (RHS) Pools:
\begin{itemize}
\item From \(X_1, X_2\): Two independent sets of exponentials at each rate \(j \in \{1, \dots, f+1\}\).
\item From \(X_{3}\): One set of exponentials at each rate \(j \in \{1, \dots, 2f+1\}\).
\end{itemize}

Left-Hand Side (LHS) Pools:
\begin{itemize}
\item From \(Y\): One standard exponential at rate \(1\).
\item From \(Y_1, Y_2, Y_4, Y_5\): Four independent sets of exponentials at each rate \(j \in \{1, \dots, f+1\}\).
\item From \(Y_{3}\): One set of exponentials at each rate \(j \in \{1, \dots, 2f+1\}\).
\end{itemize}

Total Multiplicities by Rate Pool:\\
By aggregating the variables across their respective spans, we establish the pole structures for the system:
\begin{itemize}
\item Rate range: \(j = 1\); LHS multiplicity (Poles at \(-j\)): 6; 
RHS multiplicity (Poles at \(+j\)): 3 
\item Rate range: \(2 \le j \le f+1\); LHS multiplicity (Poles at \(-j\)) : 5; 
RHS multiplicity (Poles at \(+j\)): 3 
\item Rate range: \(f+2 \le j \le 2f+1\); LHS multiplicity (Poles at \(-j\)): 1; 
RHS multiplicity (Poles at \(+j\)): 1 
\end{itemize}

Letting \(Z = \text{RHS} - \text{LHS}\), the goal is to evaluate \(P(Z > 0)\). The Moment Generating Function (MGF) of \(Z\) is:\\
\(M_{Z}(s)=\left[\prod _{j=1}^{2f+1}\left(\frac{j}{j-s}\right)^{m_{\text{RHS}}(j)}\right]\times \left[\prod _{j=1}^{2f+1}\left(\frac{j}{j+s}\right)^{m_{\text{LHS}}(j)}\right]\)

By closing the complex contour around the Right Half-Plane (RHP), we isolate only the positive poles. This minimizes the computation since the higher-index positive poles drop down to a multiplicity of 1. 
The exact probability is found by summing the residues of \(\frac{M_{Z}(s)}{s}\) at these positive poles:
\(P(Z>0)=-\sum _{j=1}^{2f+1}\text{Res}\left(\frac{M_{Z}(s)}{s},s=j\right)\)
\begin{itemize}
\item For \(j \le f+1\): The poles have a multiplicity of 3, requiring a 2nd-order derivative to find the residue.
\item For \(j \ge f+2\): The poles have a multiplicity of 1, allowing for a direct substitution without derivatives.
\end{itemize}

To find the asymptotic behavior of the probability, 
we analyze how the logarithmic growth of the maximum statistics scales on both sides of the inequality. 
%
As the sample size grows large, 
the maximum of \(M\) independent standard exponential random variables converges to a Gumbel distribution shifted by \(\ln M\):
\(\text{Max}_{M}\approx \ln M+G\),
where \(G\) is a standard Gumbel random variable with an exponential upper tail \(P(G > x) \approx e^{-x}\).
\begin{enumerate}
\item Right-Hand Side (RHS):
The RHS contains two variables of size \(f+1\) and one of size \(2f+1\):
\(\text{RHS}\approx \ln (f+1)+\ln (f+1)+\ln (2f+1)+(G_{X,1}+G_{X,2}+G_{X,3})\)
Using the logarithmic property \(\ln(2f+1) \approx \ln f + \ln 2\), this simplifies to:
\(\text{RHS}\approx 3\ln f+\ln 2+\sum _{i=1}^{3}G_{X,i}\)
\item Left-Hand Side (LHS): The LHS contains four variables of size \(f+1\), one of size \(2f+1\), and the standard exponential variable \(Y\):
\(\text{LHS}\approx 4\ln (f+1)+\ln (2f+1)+Y+\sum _{j=1}^{5}G_{Y,j}\)
Simplifying the logarithmic growth:
\(\text{LHS}\approx 5\ln f+\ln 2+Y+\sum _{j=1}^{5}G_{Y,j}\)
\end{enumerate}

Determining the Big-O Bound:\\
We want to find the probability that \(\text{RHS} - \text{LHS} > 0\). Subtracting the two expressions gives:\\
\(\left(3\ln f+\ln 2+\sum _{i=1}^{3}G_{X,i}\right)-\left(5\ln f+\ln 2+Y+\sum _{j=1}^{5}G_{Y,j}\right)>0\)

Notice that the \(\ln 2\) terms from the \(2f+1\) distributions cancel out perfectly, leaving:
\(\sum _{i=1}^{3}G_{X,i}-\sum _{j=1}^{5}G_{Y,j}-Y>2\ln f\).
Let \(W\) represent the combined stochastic error terms: \(W = \sum_{i=1}^3 G_{X,i} - \sum_{j=1}^5 G_{Y,j} - Y\). The problem reduces to evaluating the upper-tail probability \(P(W > 2\ln f)\):
\begin{enumerate}
\item The Exponential Decline (\(f^{2}\) Denominator):
The systemic mean shift moving against the inequality grows at a rate of \(2\ln f\). Because the upper tails of Gumbel variables decay exponentially (\(e^{-x}\)), evaluating this decay at the required shift threshold yields:
\(e^{-2\ln f}=\frac{1}{f^{2}}\)
\item The Polynomial Prefactor (\(\ln ^{2}f\) Numerator):
The upper tail of \(W\) is entirely driven by the convolution of the 3 positive Gumbel variables (\(G_{X,1}, G_{X,2}, G_{X,3}\)). Convolving \(k\) independent variables with exponential tails introduces a polynomial prefactor of order \(x^{k-1}\). For \(k=3\), the tail profile behaves as \(\mathcal{O}(x^2 e^{-x})\). Substituting \(x = 2\ln f\) yields:
\((2\ln f)^{2}=\mathcal{O}(\ln ^{2}f)\)
\end{enumerate}

Combining these two dynamics gives the final tight asymptotic bound:
\(P(Z>0)=\mathcal{O}\left(\frac{\ln ^{2}f}{f^{2}}\right) = P(B<A)\).

If we remove \(X_{3}\) from the Right-Hand Side (RHS), proceeding along similar lines, it follows that the probability scales asymptotically as \(\mathcal{O}\left(\frac{\ln f}{f^3}\right) = P(B<A')\).

If we remove \(Y_{5}\) from the Left-Hand Side (LHS) while keeping \(X_{3}\) removed from the RHS, proceeding along similar lines, it follows that the probability scales asymptotically as \(\mathcal{O}\left(\frac{\ln f}{f^2}\right) = P(B'<A')\).

Table~\ref{tab:probcomparison} compares features of the three computations.
\end{proof}

\begin{theorem}[$P(D<C)$] $P(D<C) = \mathcal{O}\left(\frac{1}{f^4} \right)$
\label{th:dlc}
\end{theorem}
\begin{proof}
For the inequality \(Y + Y_2 + Y_3 + Y_4 + Y_5 < X\), where \(X\) is an independent standard exponential random variable, we show that the probability scales asymptotically as \(\mathcal{O}\left(\frac{1}{f^4}\right)\).
Because \(X\) is a single standard exponential variable on the larger side of the inequality, this problem can be solved exactly for any \(f\) without needing Gumbel approximations.

Let \(V = Y + Y_2 + Y_3 + Y_4 + Y_5\). Because \(V\) is a sum of strictly positive random variables completely independent of \(X\), we can find the probability \(P(X > V)\) by conditioning on \(V\):
\(P(X>V)=\mathbb{E}[P(X>V\mid V)]\).
Since \(X \sim \text{Exp}(1)\), its survival function is exactly \(P(X > v) = e^{-v}\). Substituting this in gives:
\(P(X>V)=\mathbb{E}[e^{-V}]=M_{V}(-1)\) where \(M_V(-1)\) is the Moment Generating Function (MGF) of \(V\) evaluated at \(s = -1\). Because all components of \(V\) are independent, the MGF of the sum is simply the product of their individual MGFs:\\
\(M_{V}(-1)=M_{Y}(-1)\cdot M_{Y_{2}}(-1)\cdot M_{Y_{4}}(-1)\cdot M_{Y_{5}}(-1)\cdot M_{Y_{3}}(-1)\)
\begin{enumerate}
\item MGF of the Standard Exponential (\(Y\))\\
For \(Y \sim \text{Exp}(1)\), the MGF is \(\frac{1}{1-s}\). Evaluated at \(s = -1\):
\(M_{Y}(-1)=\frac{1}{1-(-1)}=\frac{1}{2}\)
\item MGF of the Maximum Statistics (\(Y_2, Y_4, Y_5, Y_3\))\\
Using Rényi’s representation, the maximum of \(K\) independent standard exponential variables has an MGF of \(\prod_{j=1}^K \frac{j}{j-s}\). When evaluated at \(s = -1\), this telescopes perfectly:
\(M_{\max _{K}}(-1)=\prod _{j=1}^{K}\frac{j}{j+1}=\frac{1}{2}\cdot \frac{2}{3}\cdot \frac{3}{4}\cdots \frac{K}{K+1}=\frac{1}{K+1}\)

Applying this clean identity to our specific variables: 
\begin{itemize}
\item For \(Y_2, Y_4, Y_5\) (each a maximum of \(K = f+1\) variables): \(M_{Y_{i}}(-1)=\frac{1}{(f+1)+1}=\frac{1}{f+2}\)
\item For \(Y_{3}\) (a maximum of \(K = 2f+1\) variables): \(M_{Y_{3}}(-1)=\frac{1}{(2f+1)+1}=\frac{1}{2f+2}\)
\end{itemize}
\end{enumerate}
Multiplying these independent expectations together yields: 
\(P(X>V)=\frac{1}{2}\cdot \left(\frac{1}{f+2}\right)^{3}\cdot \frac{1}{2(f+1)}=\frac{1}{4(f+1)(f+2)^{3}} = P(D<C)\)
\end{proof}

\subsection{Analytical Evaluation of Causal Weak Safety Violation ($\mathcal{O}(\frac{\ln^3 f}{f^{4}})$)}
\label{sec:analytical_eval}
\begin{theorem}[Causal Weak Safety Violation Bound]
Under optimal resilience bounds ($n \ge 3f + 1$), the probability of a causal ordering weak safety violation at any correct process $p_j$ (for a path length of 2: $p_i$ to $p_k$ to $p_j$) is bounded by 
$\mathcal{O}(\frac{\ln^3 f}{f^{4}})$.
\label{th:ws2analysis}
\end{theorem}
\begin{proof}
Using the bounds derived in Theorems~\ref{th:bla},~\ref{th:dlc}, the probability of a causal safety violation at process $p_j$ as per Theorem~\ref{th:wsformulation} is:
\[
P(B < A) \times [1 - (1 - P(B < A')) \times (1 - P(D < C))]
\]
\[
\leq \mathcal{O}(\frac{\ln^2 f}{f^{2}}) \times [1 - (1 - \mathcal{O}(\frac{\ln f}{f^{3}})) \times (1 - \mathcal{O}(\frac{1}{f^{4}}))] \leq \mathcal{O}(\frac{\ln^3 f}{f^{5}})
\]
Since $p_k$ could be any correct process, by a union bound over the $n - f \sim 2f$ correct processes:
\[
P(\text{CO violation}) \leq 2f \cdot \mathcal{O}(\frac{\ln^3 f}{f^{5}}) \leq \mathcal{O}(\frac{\ln^3 f}{f^{4}}) 
\]
\end{proof}

\subsection{Causal Weak Safety Violation Probability for All Path Lengths ($\mathcal{O}(\frac{\ln^3 f}{f^{4}})$)}
We generalize the analysis to causal chains of longer path lengths through non-repeating/unique processes.
\begin{theorem}[Causal Weak Safety Violation for Longer Paths]
Under optimal resiliency bounds ($n \ge 3f+1$), for all causal chains of path length $L$, where $(n-f) > L > 2$, passing through correct processes, the sum of the probabilities of a causal ordering weak safety violation at any correct process is strictly smaller than that of path length 2, which is $\mathcal{O}(\frac{\ln^3 f}{f^{4}})$, 
as the system size scales.
\label{th:ws3+analysis}
\end{theorem}
\begin{proof}
Let a causal chain for path B of length $L \ge 2$ be $m_1 \to m_2 \to \dots \to m_L$, where the messages are broadcasted sequentially by correct processes. 
Path A (BCRB-delivery of $m_1$ to $p_j$) is independent of the chain length $L$. Path B (BCRB-delivery of $m_L$ to $p_j$) scales with $L$: 
\[
    B(L) = T_1^{(2)} + T_2^{(2)} + L\left(Y_1(m_2, \text{INIT}, 1, -) + T_3^{(2)} + T_4^{(2)} + T_5^{(2)}\right)
    \]
    \[
    = Y_1 + Y_2 + L\left(Y + Y_3 + Y_4 + Y_5\right)
    \]

We first compute $P(B<A)$ for the asymptotic case.
%
Using the Gumbel approximation (\(\text{Max}_M \approx \ln M + G\)), we map out the growth rates of both sides:
\begin{itemize}
\item Max of \(f+1\) variables: \(Y_1, Y_2, Y_4, Y_5, X_1, X_2 \approx \ln f + G\)
\item Max of \(2f+1\) variables: \(Y_3, X_3 \approx \ln(2f) + G = \ln f + \ln 2 + G\)
\item Standard exponential variable: \(Y \approx \mathcal{O}(1)\)
\end{itemize} 
\begin{enumerate}
\item Right-Hand Side (RHS) Growth:
The RHS contains two variables of size \(f+1\) and one variable of size \(2f+1\):
\(\text{RHS}\approx \ln f+\ln f+(\ln f+\ln 2)+\sum _{i=1}^{3}G_{X,i}=3\ln f+\ln 2+\sum _{i=1}^{3}G_{X,i}\)
\item  Left-Hand Side (LHS) Growth:
The modified LHS scales \(Y, Y_3, Y_4, Y_5\) by the factor \(L\):
\(\text{LHS}\approx (\ln f+G_{Y,1})+(\ln f+G_{Y,2})+L\Big(\mathcal{O}(1)+(\ln f+\ln 2+G_{Y,3})+\ln f+\ln f+G_{Y,4}+G_{Y,5}\Big)\)

Collapsing the logarithmic coefficients yields:
\[\text{LHS}\approx 2\ln f+G_{Y,1}+G_{Y,2}+L\big(3\ln f+\ln 2+G_{Y,3}+G_{Y,4}+G_{Y,5}\big)\]
\[\text{LHS}\approx (2+3L)\ln f+L\ln 2+G_{Y,1}+G_{Y,2}+L(G_{Y,3}+G_{Y,4}+G_{Y,5})\]
\end{enumerate}

We now evaluate \(P(\text{RHS} - \text{LHS}) > 0\): 
\(\left(3\ln f+\ln 2+\sum _{i=1}^{3}G_{X,i}\right)-\-\left((2+3L)\ln f+L\ln 2+G_{Y,1}+G_{Y,2}+\-L(G_{Y,3}+G_{Y,4}+G_{Y,5})\right)>0\).
Isolating the deterministic growth terms on the right-hand side gives:
\[\sum _{i=1}^{3}G_{X,i}-G_{Y,1}-G_{Y,2}-L(G_{Y,3}+G_{Y,4}+G_{Y,5})>(3L-1)\ln f+(L-1)\ln 2\]
Let \(W\) represent the combined stochastic terms on the left. The upper-tail probability \(P(W > (3L - 1)\ln f + (L-1)\ln 2 
)\) decomposes into:
\begin{enumerate}
\item The Exponential Decline (\(f^{3L-1}\) Denominator):
Because \(L > 1\), the LHS outpaces the RHS. The net deterministic drag moving against the inequality grows at a rate \(\approx\) \((3L - 1)\ln f\). Evaluating the exponential tail decay (\(e^{-x}\)) at this exact threshold yields:
\(e^{-(3L-1)\ln f}=\frac{1}{f^{3L-1}}\)
\item The Polynomial Prefactor (\(\ln ^{2}f\) Numerator):
The upper tail behavior of \(W\) is entirely governed by the convolution of the 3 positive Gumbel variables on the RHS (\(G_{X,1}, G_{X,2}, G_{X,3}\)). Convolving \(k=3\) independent variables with exponential tails introduces a polynomial prefactor of order \(x^{k-1} = x^2\). Substituting the threshold \(x \approx (3L - 1)\ln f\) results in:
\(((3L-1)\ln f)^{2}=\mathcal{O}(L^2 \ln ^{2}f)\)
\end{enumerate}
Combining these two yields the corrected, final asymptotic bound:
\(P(Z>0)=\mathcal{O}\left(\frac{L^2 \ln ^{2}f}{f^{3L-1}}\right) = P(B<A)\)

Path A' (BRB-delivery of $m_1$ to $p_j$) is independent of the chain length $L$. Similar to the derivation of $P(B < A)$ above, we have:
\(P(B<A') = \mathcal{O}\left(\frac{L \ln f}{f^{3L}}\right)\)

For the $P(D < C)$ term, the delivery race condition  concerns each suffix of the causal chain ($m_{L-x} \to m_L$). Based on the derivation of $P(D < C)$ of the baseline case of path length 2, we have:
\[
P(D < C) = \sum^{L-1}_{x=2}\frac{1}{(4(f+1)(f+2)^3)^x} \le \approx \mathcal{O}(f^{-4})
\]
To perform a worst-case analysis considering multiple paths of length $L$, we note that there are at most $(n-f)^{L-1} \approx (2f)^{L-1}$ causal paths of length $L$ for $B$. (This assumes $n = \mathcal{O}(f)$ and close to optimal resilience bound $3f+1$.) By taking a union bound over these paths, the probability that any such path $B(L)$ delivers before path $A$ and violates gating is bounded by:
\[
P(\text{CO violation for all paths of length } L)
\]
\[
\le (2f)^{L-1} \cdot P(B < A) \cdot \left[ 1 - (1 - P(B < A')) \cdot (1 - P(D < C))\right]
\]
\[
\le (2f)^{L-1} \cdot \mathcal{O}(\frac{L^2 \ln^2 f}{f^{3L-1}}) \cdot \left(1 - (1 - \mathcal{O}(\frac{L \ln f}{f^{3L}})) \cdot (1 - \mathcal{O}(f^{-4})\right)
\]
\[
\le (2f)^{L-1} \cdot \mathcal{O}(\frac{L^2 \ln^2 f}{f^{3L-1}}) \cdot \left(\mathcal{O}(\frac{L \ln f}{f^{3L}}) + \mathcal{O}(f^{-4})\right)
\]
To evaluate the overall weak safety violation probability for paths of any length greater than 2, we sum the probabilities for all possible path lengths $L$ from $3$ to the maximum path length through correct processes, $n-f \sim 2f$:
\[
\sum_{L=3}^{2f} (2f)^{L-1} \cdot \mathcal{O}(\frac{L^2 \ln^2 f}{f^{3L-1}}) \cdot \left(\mathcal{O}(\frac{L \ln f}{f^{3L}}) + \mathcal{O}(f^{-4})\right)
\]
\[
\le \sum_{L=3}^{2f} (2f)^{L-1} \cdot \mathcal{O}(\frac{L^2 \ln^2 f}{f^{3L-1}}) \cdot \left(\mathcal{O}(f^{-4})\right) \leq \mathcal{O}(\frac{\ln^2 f}{f^4})
\]
This confirms that longer paths lead to strictly lower order violation probabilities than for paths of length 2 (Theorem~\ref{th:ws2analysis}).
\end{proof}

\begin{corollary}
\label{co:wstotalprobability}
The probability of weak safety violation (without considering deadlock avoidance) over all path lengths $\geq 2$ is $\mathcal{O}(\frac{\ln^3 f}{f^{4}})$.
\end{corollary}

\subsection{Analytical Evaluation of Strong Safety Violation: Algorithm~\ref{alg:bcrb_algorithm},~\ref{alg:bcrb_algorithm_n3} ($\mathcal{O}(\frac{\ln^2 f}{f})$)}

\begin{theorem}[Probabilistic (Strong) Causal Safety]
Let process $p_i$ broadcast message $m_1$, and (correct or Byzantine) process $p_k$ broadcast message $m_2$ after delivering $m_1$. The probability that any correct process $p_j$ delivers $m_2$ before $m_1$ is 
\[
P(\text{CO violation at } p_j) = P(B < A)
\]
where $A$, $B$ are defined in the body of the proof of Theorem~\ref{th:wsformulation}.
\label{th:ssformulation}
\end{theorem}
\begin{proof}
Follows from the proof of Theorem~\ref{th:wsformulation}. The only difference is that a Byzantine process $p_k$ may not send ACKs. So $C = \infty$ and $P(D < C)$ = 1. 
\end{proof}

\begin{theorem}[Causal Strong Safety Violation Bound]
Under optimal resilience bounds ($n \ge 3f + 1$), the probability of a causal ordering strong safety violation at any correct process $p_j$ (for a path length of 2: $p_i$ to $p_k$ to $p_j$) is bounded by 
$\mathcal{O}\left(\frac{\ln^2 f}{f}\right)$.
\label{th:ss2analysis}
\end{theorem}
\begin{proof}
Follows from the proof of Theorems~\ref{th:ssformulation} and~\ref{th:bla}. 
\(P(B < A) \leq \mathcal{O}(\frac{\ln^2 f}{f^{2}})
\)
for the path $p_i$ to $p_k$ to $p_j$ but considering all $(n-2) \sim 3f (= \mathcal{O}(f))$ paths of length of 2 from $p_i$ to $p_j$:
\[P(B < A) \leq (3f) \times \mathcal{O}\left(\frac{\ln^2 f}{f^{2}}\right) \leq \mathcal{O}\left(\frac{\ln^2 f}{f}\right)
\]
\end{proof}

\begin{theorem}[Causal Strong Safety Violation for Longer Paths]
Under optimal resiliency bounds ($n \ge 3f+1$), for all causal chains of path length $L$, where $n > L > 2$, passing through correct or Byzantine processes, the sum of the probabilities of a causal ordering strong safety violation at any correct process is strictly smaller than that of path length 2, which is $\mathcal{O}(\frac{\ln^2 f}{f})$,  as the system size scales.
\label{th:ss3+analysis}
\end{theorem}
\begin{proof}
Follows from the proofs of Theorem~\ref{th:ssformulation} and similar to the proof of Theorem~\ref{th:ws3+analysis}. 
For a path of length $L$, assume the worst-case scenario that all processes on the path are Byzantine (and even when $L > f$), and none of them send ACKs. Then $P(D < C) = 1$, resulting in the highest probability of CO violation as per Theorem~\ref{th:ssformulation} formula.  The number of paths of length $L$ from $p_i$ to $p_j$ is $n^{L-1} \sim (3f)^{L-1}$ (assuming $n = \mathcal{O}(f)$).
\[
P(\text{CO violation for all paths of length } L) \le (3f)^{L-1} \cdot P(B < A)
\]
\[
\le (3f)^{L-1} \cdot \mathcal{O}(\frac{L^2 \ln^2 f}{f^{3L-1}}) = \mathcal{O}(\frac{L^2 \ln^2 f}{f^{L+1}})
\]
To evaluate the overall strong safety violation probability for paths of all lengths greater than 2, we sum the probabilities for all possible path lengths $L$ from $3$ to the maximum length of paths through unique processes, $n \sim 3f$:
\[\sum_{L=3}^{3f} \mathcal{O}(\frac{L^2 \ln^2 f}{f^{L+1}}) \le \mathcal{O}(\frac{\ln^2 f}{f^2})
\]
The sum is certainly $\mathcal{O}(\frac{\ln^2 f}{f})$, which is the probability of violation for path length 2 (Theorem~\ref{th:ss2analysis}), confirming that longer paths lead to strictly lower order violation probabilities.
\end{proof}

\begin{corollary}
\label{co:sstotalprobability}
The probability of strong safety violation over all path lengths $\geq 2$ is $\mathcal{O}(\frac{\ln^2 f}{f})$.
\end{corollary}

\subsection{Impact of Deadlock Avoidance on Weak Safety Violation ($P = \mathcal{O}(\frac{\ln^3 f}{f^{3}})$)}
\label{sec:deadlock}
\begin{theorem}
\label{th:wsdeadlock}
[Causal Weak Safety Violation Bound Considering Deadlock Avoidance]
Under optimal resilience bounds ($n \ge 3f + 1$), the probability of a causal ordering weak safety violation at any correct process $p_j$, considering impact of deadlock avoidance, 
is bounded by:
\[
P(\text{CO violation at some correct process}) \le \mathcal{O}(\frac{\ln^3 f}{f^{3}})
\]
\end{theorem}
\begin{proof}
A deadlock cycle may arise if there is (at least) one Byzantine process $p_b$ that sends $ACK(m=(l,s), next\_sn, \ldots)$ where it has already BCRB-broadcast  $(m_b=(b,next\_sn'), pl)$ and $next\_sn' \geq next\_sn$, and there is a (possibly transitive) causal dependency of $m$ on $m_b$ (i.e., there is a path from $m$ to $m_b$ in $G$) because $m_b$ is BRB-delivered and processed at $p_l$ before $p_l$ BCRB-broadcast $m$. Thus there is a path in $G$ from $m_b$ to $m$ (due to edge or path $(m_b,m)$) and from $m$ to $m_b$, thereby completing a cycle. 

For a true dependency edge $((z,s_z),(y,s_y))$ to form, the physical time of BRB-broadcast($(z,s_z)$) $>$ physical time of BRB-broadcast($(y,s_y)$). However, a third observer process $p_j$ may observe these two broadcasts in any order based on the completion times of the BRB-deliveries at $p_j$, and this can lead to observing a cycle. Algorithm~\ref{alg:bcrb_algorithm} uses deadlock avoidance: it assigns local physical clock timestamps ($.ts$) to nodes in $G$ at the time they are locally formed (BRB-delivery processed); an edge $((z,s_z),(y,s_y))$ is not added to $G$ if $(z,s_z).ts < (y,s_y).ts$. This guarantees that $G$ is always a DAG.  

To analyze the impact of the non-addition of a correct edge on the probability of weak safety violation in a worst-case scenario, it is sufficient to analyze the following scenario, also considered in the setting of Theorem~\ref{th:wsformulation}. Let correct process $p_i$ BCRB-broadcast message $m_1$, and correct process $p_k$ BCRB-broadcast message $m_2$ after BCRB-delivering $m_1$. The true dependency edge $(m_2,m_1)$ is not added to $G$ at correct observer $p_j$ if $m_2.ts < m_1.ts$. The probability of this happening is $P(B'<A')$, where $B',A'$ were defined in the proof of Theorem~\ref{th:wsformulation} and $P(B'<A') = \mathcal{O}(\frac{\ln f}{f^2})$ was computed in Theorem~\ref{th:bla}.

Therefore, the probability of weak safety violation considering non-addition of edge $(m_2,m_1)$ at $p_j$ is:
\[
P(WSVDA) = P(\text{weak safety violation with deadlock avoidance})
\]
\[
= P(B<A)\cdot P(B'<A')  + P( stmt. \; of \; Theorem~\ref{th:wsformulation}) \cdot (1- P(B'<A'))  
\]
\[= \mathcal{O}(\frac{\ln^2 f}{f^{2}}) \cdot \mathcal{O}(\frac{\ln f}{f^{2}}) + \mathcal{O}(\frac{\ln^3 f}{f^{5}}) \cdot (1 - \mathcal{O}(\frac{\ln f}{f^{2}})) = \mathcal{O}(\frac{\ln^3 f}{f^{4}})
\]
As $p_k$ in our scenario could be any of $n-f \simeq 2f$ correct processes, using union-bound over $f$ processes:
\[
P(\text{WSVDA}) = 2f\cdot \mathcal{O}(\frac{\ln^3 f}{f^{4}}) = \mathcal{O}(\frac{\ln^3 f}{f^{3}})
\]
Thus the probability of weak safety violation when considering deadlock avoidance increases from $\mathcal{O}(\frac{\ln^3 f}{f^{4}})$ to $\mathcal{O}(\frac{\ln^3 f}{f^{3}})$. 
\end{proof}
The probability of strong safety violation even considering deadlock avoidance remains unaffected at $\mathcal{O}(\frac{\ln^2 f}{f})$.

\subsection{Validity, Agreement, and Integrity}
\label{sec:vai}
\begin{theorem}
\label{th:ordering}
If a correct process $p_i$ adds edge $((z, next\_sn), (M' = (x,s_x)))$ in $G$, then $M'$ has been BRB-broadcast by $p_x$ before broadcast of $ACK(M', next\_sn, share, z, h)$ by $p_z$. 
\end{theorem}
\begin{proof}
If $p_z$ is a correct process, the theorem follows from the Algorithm pseudo-code. So consider that $p_z$ is Byzantine. 

In receive ACK processing, the hash of $M'.ciphertext$ sent on ACK as parameter $h$ must match the hash $M'.ciphertext$ as computed by $p_i$ on BRB-delivery of $M'$. This ensures $M'$ was BRB-broadcast before broadcast of $ACK(M')$ in order for $p_z$ to compute and piggyback the correct hash of $M'$ on $ACK(m')$. 
\end{proof}

\begin{theorem}[BCRB Delivery]
\label{th:bcrb_agreement_standalone}
If a message $m_1 = ((b_1, s_1), \text{payload})$ from process $b_1$ (whether $b_1$ is a Byzantine or correct process) is inserted into the local $\texttt{pending}$ set of any correct process $p_i$ 
then $m_1$ will eventually be BCRB-delivered by every correct process.
\end{theorem}
\begin{proof}
If message $m_1 = ((b_1, s_1), ciphertext)$ is 
inserted into $\texttt{pending}$ at any correct process $p_i$ (implying $m_1$ was BRB-delivered at $p_i$), then all  $m_{1'}$ = $(b_1, s_{1'}) | s_{1'} \leq s_1$, are also BRB-delivered locally and inserted in $pending$. By the Agreement property of BRB layer, $m_1$ and all $m_{1'}$ are also eventually BRB-delivered at every correct process $p_k$ and added to $p_k$'s local $\texttt{pending}$ set. 

Upon BRB delivery of $m_1$, every correct process $p_k$ broadcasts an $ACK(m_1, next\_sn, \-share, k, h)$. 
As there are at least $n - f$ correct processes, every correct process $p_k$ eventually receives at least $n - f$ ACKs/shares for $m_1$ (and $m_{1'}$), satisfying $quorum\_ok$. 

As $G$ is acyclic (and hence a DAG), all paths from $m_1$ lead to leaf nodes $m=(b,s_b)$ at any point in time. $m$ satisfies $fifo\_ok$ because all lower seq numbered messages (from $b$) have been BCRB-delivered and removed from $G$. Within bounded time $m$ will satisfy $quorum\_ok$ by the reasoning in the above para. Thus it will be BCRB-delivered and $m$ and its incoming edges deleted from $G$ leading to shorter paths to new leaf nodes. However if $quorum\_ok$ is not satisfied at this point in time, an outgoing edge $(m,m')$ may be added on receipt of ACK($m',s_b, share, b, h$), provided the earlier sent $m'$ than $m$ (by Theorem~\ref{th:ordering}) is BRB-delivered and processed but not yet BCRB-delivered. This leads to a new leaf node that satisfies $fifo\_ok$ and guaranteed to satisfy $quorum\_ok$. Such extensions to the path(s) from $m$ are bounded by the size of the history, which contains a bounded number of messages. Eventually the paths from $m$ must shrink and $m$ will be BCRB-delivered and along with its incoming edges removed from $G$.

The same logic shows that the new sink nodes reachable via shorter paths from $m_1$ will be BCRB-delivered and the paths from $m_1$ will inductively keep shrinking until $m_1$ gets BCRB-delivered.
\end{proof}

\begin{theorem}
Algorithms~\ref{alg:bcrb_algorithm},~\ref{alg:bcrb_non_crypto_algorithm} satisfy Validity, Agreement, and Integrity at the BCRB layer.
\end{theorem}
\begin{proof}
Follows from Theorem~\ref{th:bcrb_agreement_standalone} and the properties of the underlying BRB layer:
\begin{itemize}
    \item \textbf{Validity:} If a correct process $p_i$ BCRB-broadcasts $m$, by Validity of the underlying BRB layer, $m$ is BRB-delivered to each correct process. It will be inserted into $\texttt{pending}$ at all correct processes as $p_i$ must have BCRB-broadcast all messages with lower sequence numbers in order, which will also be BRB-delivered in order and placed in their $\texttt{pending}$ at all correct processes. By Theorem~\ref{th:bcrb_agreement_standalone}, $m$ is eventually BCRB-delivered by every correct process.
    \item \textbf{Agreement:} If a correct process $p_i$ BCRB-delivers $m$, $m$ was BRB-delivered and placed in $\texttt{pending}$ at $p_i$. By Theorem~\ref{th:bcrb_agreement_standalone}, $m$ is eventually BCRB-delivered by every correct process.
    \item \textbf{Integrity:} At-most-once delivery is guaranteed because $\mathbf{V}_k[j]$ strictly increments upon delivery, causing $fifo\_ok$ to evaluate to \texttt{False} for any duplicate processing. Authenticity for correct senders follows from the Integrity property of the underlying BRB layer: since correct processes only send messages via \texttt{bcrb\_broadcast}, and the underlying BRB layer guarantees sender authenticity, a Byzantine process directly invoking \texttt{brb\_broadcast} (bypassing the BCRB layer) cannot forge messages from a correct sender.
\end{itemize}
\end{proof}

A Byzantine sender can broadcast conflicting 
messages with arbitrary sequence numbers or omit
sending ACKs selectively. This can at worst result in the Byzantine sender’s own message stream being blocked or halted permanently at correct processes. The impact on safety was already accounted for in violation bounds.

\section{Weak Safety Guarantee via BRB-Gated ACKs}
\label{sec:detweaksafety}
\begin{algorithm*}[tp]
\small
\SetKwComment{Comment}{// \textit{}}{}
\SetKwFunction{BCRBBroadcast}{bcrb\_broadcast}
\SetKwFunction{BRBBroadcast}{brb\_broadcast}
\SetKwFunction{CheckDelivery}{check\_delivery}
\SetKwFunction{BRBDeliver}{brb\_deliver}
\SetKwFunction{ReceiveACK}{receive ACK}
\SetKwProg{Fn}{procedure}{:}{}
\SetKwProg{On}{upon}{:}{}

\caption{Byzantine Causal Reliable Broadcast (BCRB) with BRB of ACKs (Process $i$)}
\label{alg:bcrb_algorithm_n3}

\begin{multicols}{2}
\SetInd{0.2em}{0.5em}

\textbf{state variables:} \\
  $local\_sn \gets 0$ {\footnotesize \Comment*{Seq. num. for process $i$'s broadcasts}}
  $\mathbf{V}_i \gets [0, \dots, 0]$ {\footnotesize \Comment*{Vec. of BCRB-delivered seq. nos.}}
  $pending \gets \emptyset$ {\footnotesize \Comment*{processed BRB-delivery, pending BCRB delivery}}
  $shares \gets \text{Array of } \emptyset$ {\footnotesize \Comment*{Maps msg ID to decryption shares}}
  $G=(N,E) \gets (\emptyset,\emptyset)$ {\footnotesize \Comment*{graph of dependencies between msg IDs}}

\BlankLine
\Fn{\BCRBBroadcast{payload}}{
    $local\_sn \gets local\_sn + 1$\;
    $ciphertext \gets \text{encrypt}(payload, \mathit{PK})$\;
    \BRBBroadcast{$(i, local\_sn, ciphertext)$}\;
}

\BlankLine
\Fn{\CheckDelivery{}}{
    $progress \gets \text{True}$\;
    \While{$progress = \text{True}$}{
        $progress \gets \text{False}$\;
        \ForEach{$m = (M = (j, sn), payload) \in pending$}{
            $fifo\_ok \gets (\mathbf{V}_i[j] = sn - 1)$\;
            $quorum\_ok \gets \text{True}$\;
            \If{$payload$ is not ACK}{
                $quorum\_ok \gets (|shares[M]| \ge n - f)$\;
            }
            \If{$fifo\_ok \land quorum\_ok \land (j,sn) \text{ has no outgoing edge in } E$}{
                $pending \gets pending \setminus \{m\}$\;
                delete $M$ from $N$ and all incident edges in $E$\;
                $\mathbf{V}_i[j] \gets sn$\;
                
                \If{payload is not an ACK}{
                    $plaintext \gets \text{decrypt}(payload, shares[M], \mathit{PK})$\;
                    $\text{bcrb\_deliver}(j, sn, plaintext)$\;
                }
                $progress \gets \text{True}$\;
                \textbf{break}\;
            }
        }
    }
}

\columnbreak

\BlankLine
\On{\BRBDeliver{j, sn, payload}}{
    \textbf{wait until } all $M = (j,s), s < sn$ are brb-delivered and execution of brb-deliver is completed\;
    \If{payload is ACK[$(M \gets (k,sn^k)), share, h]$}{
       \textbf{wait until } all $M' = (k,s'), s' \le sn^k$ are brb-delivered and execution of brb-deliver is completed\;
       \If{$M \in N \land h = hash(M.ciphertext)$}{
            $shares[M] \gets shares[M] \cup \{share\}$\;
            $E \gets E \cup \{((j,sn+1)^{temp},(k,sn^k))\}$\;
       }
    }
    \Else{
        $local\_sn \gets local\_sn + 1$\;
        $share \gets \text{dec\_share}(payload, \mathit{SK}_i)$\;
        $h \gets hash(payload)$\;
        \BRBBroadcast{i, local\_sn, ACK[(j,sn),  share, h]}\; 
    }
    \If{$(j,sn)^{temp}$ \text{ exists }}{
        relabel it as $(j,sn)$\;
    }
    \Else{
        add $(j,sn)$ to $N$\;
    }
    \If{$(j, sn-1) \in N$}{
        add edge $((j,sn),(j,sn-1))$ to $E$\;
    }
    $pending \gets pending \cup \{((j, sn), payload)\}$\;
    \CheckDelivery{}\;
}

\end{multicols}
\end{algorithm*}

\begin{algorithm}[th]
\caption{Non-Cryptographic Byzantine Causal Reliable Broadcast (BCRB) with BRB of ACKs}
\label{alg:bcrb_non_crypto_algorithm_n3}
\small
Same as Algorithm~\ref{alg:bcrb_algorithm_n3} except: (1) ciphertext is identical to plaintext, i.e., encryption/decryption are idempotent operations, and (2) decryption share is set to the sender process ID. \\
\end{algorithm}

Algorithms~\ref{alg:bcrb_algorithm_n3},~\ref{alg:bcrb_non_crypto_algorithm_n3} are like Algorithms~\ref{alg:bcrb_algorithm},~\ref{alg:bcrb_non_crypto_algorithm}, resp.. Difference is that rather than sending ACKs point-to-point, they are sent via BRB. Specifically, before a correct $p_k$ BCRB-broadcasts a subsequent application message $m_2$ (and hence BRB-broadcast it), $p_k$ initiates the BRB of the $ACK(m_1)$ for the predecessor message $m_1$ that has been BRB-delivered and that processing completed.
On BRB-delivery of the ACK, it is subject to $fifo\_ok$ check.
This $fifo\_ok$ check ensures that BRBs (of ACK($m_1$) and of $m_2$) from $p_k$ are processed in FIFO order (globally) and $m_2$ is BCRB-delivered after BRB-delivered ACK($m_1$) is processed and 
the dependency of $m_2$ on $m_1$ is registered.
This guarantees 100\% weak safety. While the message overhead remains constant-size $\mathcal{O}(1)$ and the probability of strong safety violation remains the same, running a BRB instance for every ACK increases the number of messages sent to $\mathcal{O}(n^3)$. Thus the   communication word complexity is $\mathcal{O}(n^3)$; the application payload still gets sent only on $\mathcal{O}(n^2)$ messages. 
Explanation is in Appendix~\ref{sec:expl_crypto_n3}. Correctness proof is in Appendix~\ref{sec:proof_crypto_n3}.

\section{Conclusions}
\label{sec:concl}
We proposed the first $\mathcal{O}(1)$ metadata overhead BCRB algorithm with $\mathcal{O}(n^2)$ communication word  complexity. We also proved bounds on the probabilities of violation of strong safety and of weak safety. 
Such bounds for our low-cost solutions are particularly attractive for non-critical applications like social networks. Our algorithms and associated safety violation bounds (even for Algorithms~\ref{alg:bcrb_algorithm_n3},~\ref{alg:bcrb_non_crypto_algorithm_n3} which are $\mathcal{O}(n^3)$ algorithm variants) represent a good trade-off against more expensive $\mathcal{O}(n^3)$ communication word complexity algorithms \cite{Cachin2001,DBLP:journals/tcs/AuvolatFRT21}. Additionally, Cachin et al. \cite{Cachin2001} uses randomized consensus and is not throughput-scalable whereas Auvolat et al. \cite{DBLP:journals/tcs/AuvolatFRT21} has a high probability of strong safety violations (not analyzed by them) because of not just front-running attacks but also other attacks like arbitrarily erasing/manipulating causal barriers. 

\sloppy

\balance
\bibliographystyle{ACM-Reference-Format}
\bibliography{references}

@article{DBLP:journals/ppl/ImbsR16,
  author       = {Damien Imbs and
                  Michel Raynal},
  title        = {Trading off {t}-Resilience for Efficiency in Asynchronous Byzantine Reliable Broadcast},
  journal      = {Parallel Process. Lett.},
  volume       = {26},
  number       = {4},
  pages        = {1650017:1--1650017:8},
  year         = {2016},
  doi          = {10.1142/S0129626416500171},
  url          = {https://doi.org/10.1142/S0129626416500171}
}

@inproceedings{misra2022causal,
  author    = {Anshuman Misra and Ajay D. Kshemkalyani},
  title     = {Causal Ordering Properties of Byzantine Reliable Broadcast Primitives},
  booktitle = {2022 IEEE 21st International Symposium on Network Computing and Applications (NCA)},
  year      = {2022},
  pages     = {115--122},
  doi       = {10.1109/NCA55306.2022.9936749}
}

@inproceedings{misra2022detecting,
  author    = {Anshuman Misra and Ajay D. Kshemkalyani},
  title     = {Detecting Causality in the Presence of Byzantine Processes: There is No Holy Grail},
  booktitle = {2022 {IEEE} 21st International Symposium on Network Computing and Applications ({NCA})},
  year      = {2022},
  pages     = {73--80},
  doi       = {10.1109/NCA57778.2022.10013644},
  url       = {https://doi.org/10.1109/NCA57778.2022.10013644}
}

@article{DBLP:journals/pc/MisraK25,
  author       = {Anshuman Misra and
                  Ajay D. Kshemkalyani},
  title        = {Byzantine-tolerant detection of causality: There is no holy grail},
  journal      = {Parallel Comput.},
  volume       = {124},
  pages        = {103136},
  year         = {2025},
  url          = {https://doi.org/10.1016/j.parco.2025.103136},
  doi          = {10.1016/J.PARCO.2025.103136},
  bibsource    = {dblp computer science bibliography, https://dblp.org}
}

@inproceedings{DBLP:conf/icdcn/MisraK23,
  author       = {Anshuman Misra and
                  Ajay D. Kshemkalyani},
  title        = {Byzantine Fault-Tolerant Causal Ordering},
  booktitle    = {24th International Conference on Distributed Computing and Networking,
                  {ICDCN} 2023, Kharagpur, India, January 4-7, 2023},
  pages        = {100--109},
  publisher    = {{ACM}},
  year         = {2023},
  url          = {https://doi.org/10.1145/3571306.3571395},
  doi          = {10.1145/3571306.3571395},
  bibsource    = {dblp computer science bibliography, https://dblp.org}
}

@inproceedings{DBLP:conf/netys/MisraK22,
  author       = {Anshuman Misra and
                  Ajay D. Kshemkalyani},
  editor       = {Mohammed{-}Amine Koulali and
                  Mira Mezini},
  title        = {Solvability of Byzantine Fault-Tolerant Causal Ordering Problems},
  booktitle    = {Networked Systems - 10th International Conference, {NETYS} 2022, Virtual
                  Event, May 17-19, 2022, Proceedings},
  series       = {Lecture Notes in Computer Science},
  volume       = {13464},
  pages        = {87--103},
  publisher    = {Springer},
  year         = {2022},
  url          = {https://doi.org/10.1007/978-3-031-17436-0\_7},
  doi          = {10.1007/978-3-031-17436-0\_7},
  bibsource    = {dblp computer science bibliography, https://dblp.org}
}

@inproceedings{shoup2000practical,
  title={Practical threshold signatures},
  author={Shoup, Victor},
  booktitle={Advances in Cryptology—EUROCRYPT 2000: International Conference on the Theory and Application of Cryptographic Techniques, Bruges, Belgium, May 14-18, 2000. Proceedings 19},
  pages={207--220},
  year={2000},
  organization={Springer}
}

@article{DBLP:journals/tpds/MisraK24,
  author       = {Anshuman Misra and
                  Ajay D. Kshemkalyani},
  title        = {Byzantine-Tolerant Causal Ordering for Unicasts, Multicasts, and Broadcasts},
  journal      = {{IEEE} Trans. Parallel Distributed Syst.},
  volume       = {35},
  number       = {5},
  pages        = {814--828},
  year         = {2024},
  url          = {https://doi.org/10.1109/TPDS.2024.3368280},
  doi          = {10.1109/TPDS.2024.3368280},
  bibsource    = {dblp computer science bibliography, https://dblp.org}
}

@article{LLclock,
	author = {Leslie Lamport},
	journal = {Commun. ACM 21, 7},
	pages = {558-565},
	title = {Time, clocks, and the ordering of events in a distributed system},
	year = {1978}}

@article{raynal1991causal,
  title={The causal ordering abstraction and a simple way to implement it},
  author={Raynal, Michel and Schiper, Andr{\'e} and Toueg, Sam},
  journal={Information processing letters},
  volume={39},
  number={6},
  pages={343--350},
  year={1991},
  publisher={Elsevier}
}

@article{KS,
	author = {Ajay D. Kshemkalyani and Mukesh Singhal},
	bibsource = {dblp computer science bibliography, https://dblp.org},
	doi = {10.1007/s004460050044},
	journal = {Distributed Comput.},
	number = {2},
	pages = {91--111},
	title = {Necessary and Sufficient Conditions on Information for Causal Message Ordering and Their Optimal Implementation},
	url = {https://doi.org/10.1007/s004460050044},
	volume = {11},
	year = {1998}
}

@InProceedings{minicast,
author="Rumreich, Laine
and Sivilotti, Paolo A. G.",
editor="Arabnia, Hamid R.
and Takata, Masami
and Deligiannidis, Leonidas
and Rivas, Pablo
and Ohue, Masahito
and Yasuo, Nobuaki",
title="Using Minicasts for Efficient Asynchronous Causal Unicast and Byzantine Tolerance",
booktitle="Parallel and Distributed Processing Techniques",
year="2025",
publisher="Springer Nature Switzerland",
address="Cham",
pages="65--81",
isbn="978-3-031-85638-9"
}

@article{DBLP:journals/iandc/Bracha87,
  author       = {Gabriel Bracha},
  title        = {Asynchronous Byzantine Agreement Protocols},
  journal      = {Inf. Comput.},
  volume       = {75},
  number       = {2},
  pages        = {130--143},
  year         = {1987},
  url          = {https://doi.org/10.1016/0890-5401(87)90054-X},
  doi          = {10.1016/0890-5401(87)90054-X},
  bibsource    = {dblp computer science bibliography, https://dblp.org}
}

@article{DBLP:journals/jacm/BrachaT85,
  author       = {Gabriel Bracha and
                  Sam Toueg},
  title        = {Asynchronous Consensus and Broadcast Protocols},
  journal      = {J. {ACM}},
  volume       = {32},
  number       = {4},
  pages        = {824--840},
  year         = {1985},
  url          = {https://doi.org/10.1145/4221.214134},  doi          = {10.1145/4221.214134},
  bibsource    = {dblp computer science bibliography, https://dblp.org}
}

@inproceedings{Cachin2001,
	address = {Berlin, Heidelberg},
	author = {Cachin, Christian and Kursawe, Klaus and Petzold, Frank and Shoup, Victor},
	booktitle = {Advances in Cryptology --- CRYPTO 2001},
	editor = {Kilian, Joe},
	isbn = {978-3-540-44647-7},
	pages = {524--541},
	publisher = {Springer Berlin Heidelberg},
	title = {Secure and Efficient Asynchronous Broadcast Protocols},
	year = {2001}}

@article{DBLP:journals/tcs/AuvolatFRT21,
  author       = {Alex Auvolat and
                  Davide Frey and
                  Michel Raynal and
                  Fran{\c{c}}ois Ta{\"{\i}}ani},
  title        = {Byzantine-tolerant causal broadcast},
  journal      = {Theor. Comput. Sci.},
  volume       = {885},
  pages        = {55--68},
  year         = {2021},
  url          = {https://doi.org/10.1016/j.tcs.2021.06.021},
  doi          = {10.1016/J.TCS.2021.06.021},
  bibsource    = {dblp computer science bibliography, https://dblp.org}
}

\appendix

\section{Explanation of Algorithm~\ref{alg:bcrb_algorithm_n3}: Cryptographic Version with BRB of ACKs}
\label{sec:expl_crypto_n3}
\subsubsection*{bcrb\_broadcast(payload)}
When the application layer broadcasts a payload, the sender increments its $local\_sn$. The payload is encrypted under the system public key $\mathit{PK}$ to obtain a $ciphertext$. A message $m = (i, local\_sn, ciphertext)$ is created and immediately broadcasted via the underlying BRB layer. The message carries no vector clocks, ensuring $\mathcal{O}(1)$ message overhead.
\subsubsection*{upon brb\_deliver(j, sn, payload)}
This procedure is the unified entry point for both application broadcasts and ACK control messages, since ACKs are broadcasted via the underlying BRB layer rather than sent point-to-point. A process waits until all preceding messages $M = (j, s)$ for $s < sn$ from process $j$ are BRB-delivered to preserve FIFO order.
\begin{itemize}
    \item \textbf{Case 1: The payload is an ACK $[(M = (k, sn^k)), share, h]$:}
    The process waits until the target message $M$ and all its predecessor messages $M' = (k, s')$ where $s' \le sn^k$ are BRB-delivered. It then verifies that $M$ is registered in the node set $N$ and the hash matches ($h = \text{hash}(M.ciphertext)$). If valid, it records the decryption share in $\text{\texttt{shares}}[M]$ and registers a dependency by adding the edge $((j, sn+1)^{temp}, (k, sn^k))$ to the dependency graph's edge set $E$.
    \item \textbf{Case 2: The payload is an application message (not an ACK):}
    The process increments $local\_sn$, computes its cryptographic decryption share of the ciphertext using its secret key share $\mathit{SK}_i$, computes the payload hash $h \gets \text{hash}(payload)$, and broadcasts its ACK via the BRB layer: $\mathtt{brb\_broadcast}(i, local\_sn, ACK[(j, sn), share, h])$.
\end{itemize}
After handling either case, the process checks if a temporary node $(j, sn)^{temp}$ exists. If so, it relabels it as the concrete node $(j, sn)$; otherwise, it adds $(j, sn)$ to the node set $N$. To enforce FIFO order, if $(j, sn-1)$ is in $N$, it adds the edge $((j, sn), (j, sn-1))$ to $E$. Finally, it buffers the message in the \texttt{pending} set and calls \texttt{check\_delivery()}.
\subsubsection*{check\_delivery()}
This procedure evaluates the \texttt{pending} set. A pending message $m = (M = (j, sn), payload)$ is 
processed as follows. If:
\begin{enumerate}
    \item FIFO order is satisfied locally for the sender $j$ ($\mathbf{V}_i[j] = \text{sn} - 1$).
    \item If the payload is not an ACK, it has gathered at least $n - f$ valid decryption shares ($|\text{\texttt{shares}}[M]| \ge n - f$).
    \item The node $(j, sn)$ has no outgoing edge in $E$ (i.e., there are no active causal dependencies blocking its delivery).
\end{enumerate}
are satisfied, the process removes $m$ from the \texttt{pending} set, deletes node $(j, sn)$ and all its incident edges from $G$, and updates $\mathbf{V}_i[j] \gets sn$. And if the payload is an application message (not an ACK), the process decrypts the ciphertext using the shares, and triggers \texttt{bcrb\_deliver} on the plaintext. Then the process loops to check if further pending messages can now be processed.

\section{Correctness Proof of Algorithm~\ref{alg:bcrb_algorithm_n3},~\ref{alg:bcrb_non_crypto_algorithm_n3}}
\label{sec:proof_crypto_n3}
We prove that Algorithm~\ref{alg:bcrb_algorithm_n3} satisfies deadlock-freedom, 100\% Weak Safety, and Validity, Agreement, and Integrity in Byzantine environments.

\subsection{Deadlock Freedom (Acyclicity of $G$)}
\begin{theorem}
\label{th:n3_acyclic}
The dependency graph $G = (N, E)$ constructed by any correct process in Algorithm~\ref{alg:bcrb_algorithm_n3} is always a directed acyclic graph (DAG).
\end{theorem}
\begin{proof}
By construction, edges in $E$ are added in two cases:
\begin{enumerate}
    \item \textbf{FIFO Edges:} Edges of the form $((j, sn), (j, sn-1))$ are added between successive messages of the same sender $j$. Since sequence numbers strictly increase, these local edges are naturally acyclic.
    \item \textbf{Causal Dependency Edges:} Edges of the form $((j, sn+1), (k, sn^k))$ are added when process $p_i$ BRB-delivers an ACK from process $j$ sent as $j$'s message sequence number $sn$, and that processing completed.
\end{enumerate}
Suppose a cycle exists in $G$, which requires a circular dependency chain $M_1 \rightarrow M_2 \rightarrow \dots \rightarrow M_m \rightarrow M_1$. For a causal dependency edge $((j, sn+1), (k, sn^k))$ to be added:
\begin{itemize}
    \item If process $j$ is correct, it only broadcasts $ACK(M)$ after $M$ has been BRB-delivered locally and that processing completed. Thus, the physical broadcast of $M$ must have occurred in real-time before the broadcast of $ACK(M)$.
    \item If process $j$ is Byzantine, it cannot construct a valid ACK for a future message $M$ because the correct process checking the ACK enforces $h = \text{hash}(M.ciphertext)$. Since future ciphertexts are generated using a randomized threshold encryption scheme, they are unpredictable, preventing the Byzantine process from computing the correct hash $h$ beforehand.
    \item 
    The {\bf wait} steps of BRB-delivery processing guarantee FIFO processing order, defined on the sender-assigned sequence number, of all broadcasts from the same sender. If a Byzantine process does not follow these rules, the dependencies it reports on target messages and subsequent messages broadcast by it may never complete processing of BRB-delivery and hence not qualify to be inserted in $\texttt{pending}$ and in $G$ at correct processes. In fact, multiple incomplete BRB-delivery instances may have cyclic dependencies (all having been been created by Byzantine processes) but these messages or their dependencies will not enter $G$.
\end{itemize}
Because all dependencies target messages that were already BRB-delivered (and that processing completed) prior to the generation of the corresponding ACK, the dependency relation is a strict partial order based on the real-time sequence of completion of processing of BRB-deliveries. Thus, no circular dependencies can form, and $G$ is guaranteed to be acyclic.
\end{proof}

\subsection{100\% Weak Safety Guarantee}
\begin{theorem}
\label{th:n3_weak_safety}
Algorithm~\ref{alg:bcrb_algorithm_n3} guarantees 100\% Weak Causal Safety.
\end{theorem}
\begin{proof}
Weak Causal Safety equivalently states that if $m_1 \to_{brb} m_2$ and the causal chain from $m_1$ to $m_2$ passes exclusively through correct processes, then no correct process triggers \texttt{bcrb\_deliver}($m_2$) before \texttt{bcrb\_deliver}($m_1$).
Suppose there is a causal chain of messages from $m_1$ to $m_2$ broadcast via the underlying BRB layer and passing exclusively through correct processes:
\[
m_1 = m^{(0)} \to_{brb} m^{(1)} \to_{brb} \dots \to_{brb} m^{(q)} = m_2
\]
where each message $m^{(r)}$ (for $0 \le r \le q$) is broadcast by a correct process $p^{(r)}$. 
Since each process $p^{(r)}$ ($1 \le r \le q$) in the causal chain is correct, it obeys the protocol and only initiates the BRB of $m^{(r)}$ after it has BRB-delivered $m^{(r-1)}$ and completed its processing, which includes broadcasting $ACK(m^{(r-1)})$ via the BRB layer. Let $sn_r$ be the sequence number of the BRB broadcast of $m^{(r)}$ by $p^{(r)}$, so the ACK for $m^{(r-1)}$ is sent by $p^{(r)}$ with sequence number $sn_r - 1$.
By the Agreement property of the underlying BRB layer, and the FIFO processing property guaranteed by the {\bf wait} step of BRB-delivery processing, every correct process $p_i$ is guaranteed to BRB-deliver the $ACK(m^{(r-1)})$ from $p^{(r)}$ before it BRB-delivers $m^{(r)}$ from $p^{(r)}$. Upon BRB-delivering $ACK(m^{(r-1)})$, $p_i$ registers the dependency edge $((p^{(r)}, sn_r)^{temp}, m^{(r-1)})$, which is relabeled to $((p^{(r)}, sn_r), m^{(r-1)})$ in $E$ when $m^{(r)}$ is BRB-delivered.
This establishes a directed dependency path in the graph $G$ from the node of $m_2$ to the node of $m_1$:
\[
m_2 \to m^{(q-1)} \to \dots \to m_1
\]
In Algorithm~\ref{alg:bcrb_algorithm_n3}, a message $m_2$ can only be delivered to the application if it has no outgoing edges in $E$. Since $G$ is acyclic (Theorem~\ref{th:n3_acyclic}) and no dependency edges are ever deleted to resolve cycles, the path from $m_2$ to $m_1$ remains intact in $G$ as long as $m_1$ is undelivered. Thus, $m_2$ has at least one outgoing edge and cannot satisfy the delivery condition. Consequently, no correct process can BCRB-deliver $m_2$ before delivering $m_1$, satisfying the definition of Weak Causal Safety.
\end{proof}

\subsection{Validity, Agreement, and Integrity}
\label{sec:vai_n3}
\begin{theorem}
\label{th:n3_ordering}
If a correct process $p_i$ adds dependency edge $((j, sn+1)^{temp}, (k, sn^k))$ (re-labeled as $((j, sn+1), (k, sn^k))$) in $G$, then the target message $M = (k, sn^k)$ must have been broadcast via \texttt{brb\_broadcast} by $p_k$ before the \texttt{brb-broadcast} of the ACK $[M, share, h]$ by $p_j$.
\end{theorem}
\begin{proof}
If $p_j$ is a correct process, the theorem follows directly from the algorithm pseudo-code. If $p_j$ is Byzantine, the condition checking hash $h = \text{hash}(M.ciphertext)$ ensures that $M$ was already broadcast via \texttt{brb\_broadcast} by $p_k$ before $p_j$ broadcasted its ACK. Otherwise, due to the unpredictable nature of randomized threshold ciphertexts, $p_j$ could not have predicted the ciphertext hash beforehand to piggyback the correct $h$ on its ACK message.
\end{proof}

\begin{theorem}[BCRB Delivery]
\label{th:n3_bcrb_agreement_standalone}
If a message $m_1 = ((b_1, s_1), \text{payload})$ from process $b_1$ (whether $b_1$ is a Byzantine or correct process) is inserted into the local $\texttt{pending}$ set of any correct process $p_i$, then $m_1$ will eventually be BCRB-delivered by every correct process (if $\text{payload}$ is not an ACK).
\end{theorem}
\begin{proof}
If message $m_1 = ((b_1, s_1), payload)$ is inserted into $\texttt{pending}$ at any correct process $p_i$ (implying $m_1$ was BRB-delivered at $p_i$), then all $m_{1'} = (b_1, s_{1'}) \mid s_{1'} \le s_1$, and all their transitive causal predecessors that have been reported on ACKs, are also BRB-delivered and their processing completed locally in source-FIFO order and added to $\texttt{pending}$ and to $G$. By the Agreement property of the underlying BRB layer, $m_1$ and all $m_{1'}$ are also eventually BRB-delivered at every correct process $p_k$ and added to $p_k$'s local $\texttt{pending}$ set and $G$. Upon BRB delivery of $m_1$, every correct process $p_k$ broadcasts its $ACK(m_1)$ via the BRB layer if $payload$ is not ACK. Since there are at least $n - f$ correct processes, every correct process $p_k$ eventually BRB-delivers at least $n - f$ ACKs/shares for $m_1$ (and $m_{1'}$), satisfying the $quorum\_ok$ condition.
Since the dependency graph $G = (N, E)$ is a DAG (Theorem~\ref{th:n3_acyclic}) with finite history, all paths from $m_1$ lead to leaf nodes. The leaf nodes satisfy the delivery conditions ($fifo\_ok$ and no outgoing edges), and are BCRB-delivered if $payload$ is not ACK. A leaf node along with its incident edges is also deleted from the graph, and deleted from $\texttt{pending}$. This recursively reduces the path lengths for the ancestor nodes. Inductively, all predecessor dependencies of $m_1$ are delivered (if their $payload$ is not ACK), and cleared from the graph. Once all outgoing dependency edges of $m_1$ are deleted, $m_1$ itself satisfies the delivery conditions and is BCRB-delivered (if $payload$ is not ACK) at every correct process.
\end{proof}

\begin{theorem}
Algorithm~\ref{alg:bcrb_algorithm_n3} satisfies Validity (Liveness), Agreement, and Integrity at the BCRB layer.
\label{th:vai_n3}
\end{theorem}
\begin{proof}
Follows from Theorem~\ref{th:n3_bcrb_agreement_standalone} and the properties of the underlying BRB layer:
\begin{itemize}
    \item \textbf{Validity:} If a correct process $p_i$ BCRB-broadcasts $m$, by Validity of the underlying BRB layer, $m$ is BRB-delivered to each correct process. It will be inserted into $\texttt{pending}$ at all correct processes as $p_i$ must have BCRB-broadcast all messages with lower sequence numbers in order, which will also be BRB-delivered in order and placed in their $\texttt{pending}$ at all correct processes. By Theorem~\ref{th:n3_bcrb_agreement_standalone}, $m$ is eventually BCRB-delivered by every correct process.
    \item \textbf{Agreement:} If a correct process $p_i$ BCRB-delivers $m$, $m$ was BRB-delivered and placed in $\texttt{pending}$ at $p_i$. By Theorem~\ref{th:n3_bcrb_agreement_standalone}, $m$ is eventually BCRB-delivered by every correct process.
    \item \textbf{Integrity:} Integrity consists of two properties: (1) at-most-once delivery, and (2) authenticity (a correct process only delivers a message from a correct sender if that sender actually broadcast it).
    For at-most-once delivery: Upon satisfying the delivery checks in \texttt{check\_delivery()}, process $p_k$ removes the message from its \texttt{pending} set, updates its local sequence vector $\mathbf{V}_k[j] \gets sn$, and triggers \texttt{bcrb\_deliver}. Since $\mathbf{V}_k[j]$ is strictly incremented to $sn$, the FIFO check $fifo\_ok$ ($\mathbf{V}_k[j] = sn - 1$) will evaluate to \texttt{False} for any duplicate sequence number, preventing duplicate deliveries.
    For authenticity: If the sender $p_j$ is correct, it obeys the protocol and only sends messages via \texttt{bcrb\_broadcast}. By the Integrity property of the underlying BRB layer, no correct process can BRB-deliver a message $((j, sn), ciphertext)$ unless $p_j$ actually invoked \texttt{brb\_broadcast}. Even if a Byzantine process attempts to directly invoke the lower-level \texttt{brb\_broadcast} to bypass the BCRB layer, it cannot forge messages from the correct process $p_j$ because the underlying BRB layer guarantees sender authenticity (e.g., using signatures or authenticated point-to-point channels). Therefore, no correct process ever BCRB-delivers a forged message from a correct sender.
\end{itemize}
\end{proof}

Observe that the above theorems~\ref{th:n3_acyclic},~\ref{th:n3_weak_safety},~\ref{th:n3_ordering},~\ref{th:n3_bcrb_agreement_standalone},~\ref{th:vai_n3} also hold for Algorithm~\ref{alg:bcrb_non_crypto_algorithm_n3}.

Algorithm~\ref{alg:bcrb_algorithm_n3} has the same probability of strong safety violation $\mathcal{O}\left(\frac{\ln^2 f}{f}\right)$ as Algorithm~\ref{alg:bcrb_algorithm} because the same analysis as in Theorem~\ref{th:ssformulation} holds and $P(B<A)$ is the same as in Theorems~\ref{th:bla} and ~\ref{th:ws3+analysis}.

Algorithms~\ref{alg:bcrb_non_crypto_algorithm},~\ref{alg:bcrb_non_crypto_algorithm_n3} are subject to strong safety violations due to front-running attacks as well as $P(B<A)$ (the same as in Theorems~\ref{th:bla} and ~\ref{th:ws3+analysis}) 
being greater than 0.

\end{document}